%% file: one-sided-tester-full.tex
\documentclass[11pt]{article}
\usepackage[english]{babel}
\usepackage{graphicx}
\usepackage{caption}
\usepackage{subcaption}
\usepackage{scrextend}
\usepackage{thmtools}
\input{preamble.tex}

\newcommand{\findpath}{{\tt FindPath}}
\newcommand{\findclique}{{\tt FindBiclique}}
\newcommand{\localsearch}{{\tt LocalSearch}}
\newcommand{\decompose}{{\tt Decompose}}
\newcommand{\isminorfree}{{\tt FindMinor}}
\newcommand{\accept}{\textsf{ACCEPT}\xspace}
\newcommand{\reject}{\textsf{REJECT}\xspace}
\newcommand{\exit}{\textsf{EXIT}\xspace}

\renewcommand{\path}[2]{P_{#1, #2}}

\newcommand{\projw}[2]{\boldsymbol{\tau}_{#1, #2}}
\newcommand{\projwp}[3]{{\tau}_{#1, #3}(#2)}

\newcommand{\prw}[3]{p_{#1, #3}(#2)}
\newcommand{\wvec}[2]{\mathbf{p}_{#1}^{#2}}
\newcommand{\prwhop}[3]{\hat{p}_{#1, #3}(#2)}
\newcommand{\trun}[4]{q^{(#4)}_{[#3], {#1}}(#2)}
\newcommand{\trunp}[4]{\hat{q}_{[#3], {#1}}^{ (#4)}(#2)}
\newcommand{\trunvec}[3]{{\boldsymbol{q}}_{[#2], {#1}}^{(#3)}}
\newcommand{\trw}[4]{\sigma_{#1, #2, #3}(#4)}
\newcommand{\trwvec}[3]{\boldsymbol{\sigma}_{#1, #2, #3}}
\newcommand{\walk}[1]{W_{#1}}
\newcommand{\walkt}[2]{W_{#1}(#2)}
\newcommand{\walkdist}[1]{\cW_#1}
\newcommand{\threshold}{5r^4}

\newcommand{\walkset}[2]{\bW^{#2}_{#1}}
\newcommand{\conpath}[2]{P_{#1, #2}}

\newcommand{\findpathexp}{\delta(i+18)/2}
\newcommand{\findpathexptwo}{\delta(i+18)}

\newcommand{\valone}{\eps^{-2}n^{35 \delta r^4}}
\newcommand{\valtwo}{\frac{1}{\eps}n^{5 \delta r^4}}
\newcommand{\numQueries}{ \Big(\eps^{-1} n^{1/2 + O(\delta r^2)} + \eps^{-4} dn^{O(\delta r^2)} + d\eps^{-\exp (2/\delta)/2\delta} \Big) }
\newcommand{\runtime}{dn^{1/2 + O(\delta r^4)} + d\eps^{-2\exp(2/\delta)/\delta}}
\newcommand{\runtimenor}{O(dn^{1/2 + \delta} + d\eps^{-2\exp(2/\delta)/\delta})}
\newcommand{\cutoff}{\eps_{\sf CUTOFF}}
\newcommand{\RS}{\texttt{KKR}}

\newcommand{\fracbadpieces}{\frac{\eps^2}{n^{30\delta r^4}}}
\newcommand{\ballprob}{\frac{\alpha}{n^{11\delta r^4}}}
\newcommand{\numLocalSearchWalks}{\frac{n^{20\delta r^4}}{\alpha}}
\newcommand{\lenLocalSearchWalks}{\frac{160n^{6\delta r^4}} {\alpha}}
\begin{document}

\title{Finding forbidden minors in sublinear time:
a $n^{1/2+o(1)}$-query one-sided tester for minor closed properties on
bounded degree graphs}
% \author{Akash Kumar \and C. Seshadhri \and Andrew Stolman

\author{Akash Kumar\thanks{Department of Computer Science, Purdue University. {\href{mailto:akumar@purdue.edu}{akumar@purdue.edu}} (Supported in part by NSF CCF-1319080.)}
\and C. Seshadhri\thanks{Department of Computer Science, University of California, Santa Cruz. {\href{mailto:sesh@ucsc.edu}{sesh@ucsc.edu}}
(Supported by NSF TRIPODS grant CCF-1740850)}
\and Andrew Stolman\thanks{Department of Computer Science, University of California, Santa Cruz. { \href{mailto:astolman@ucsc.edu}{astolman@ucsc.edu}} (Supported by NSF TRIPODS grant CCF-1740850)}
}

% \author[1]{Akash Kumar \thanks}
% \author[2]{C. Seshadhri \thanks}
% \author[2]{ Andrew Stolman \thanks{ \href{mailto:astolman@ucsc.edu}{astolman@ucsc.edu}}}
% \affil[1]{}
% \affil[2]{}

\begin{titlepage}

\date{}
\maketitle
\thispagestyle{empty}
\abstract{Let $G$ be an undirected, bounded degree graph
with $n$ vertices. Fix a finite graph $H$, and suppose
one must remove $\varepsilon n$ edges from $G$ to make it $H$-minor free (for some small constant $\varepsilon > 0$).
We give an $n^{1/2+o(1)}$-time randomized procedure that,
with high probability, finds an $H$-minor in such a graph. 
As an application, suppose one must remove $\varepsilon n$ edges from a bounded degree graph $G$ to make it planar.
This result implies an algorithm, with the same running time, that produces
a $K_{3,3}$ or $K_5$ minor in $G$.
No prior sublinear time bound was known for this problem.

By the graph minor theorem, we get an analogous result for any minor-closed property.
Up to $n^{o(1)}$ factors, this resolves a conjecture of Benjamini-Schramm-Shapira (STOC 2008)
on the existence of one-sided property testers for minor-closed properties. Furthermore, our algorithm is nearly
optimal, by an $\Omega(\sqrt{n})$ lower bound of Czumaj et al (RSA 2014).

Prior to this work,
the only graphs $H$ for which non-trivial one-sided property testers were known for $H$-minor freeness
are the following: $H$ being a forest or a cycle (Czumaj et al, RSA 2014), $K_{2,k}$, $(k\times 2)$-grid, and the $k$-circus (Fichtenberger et al, Arxiv 2017).
}

\end{titlepage}

\section{Introduction} \label{sec:intro} 
Deciding if an $n$-vertex graph $G$ is planar is a classic algorithmic problem
solvable in linear time \cite{HT74}. The Kuratowski-Wagner theorem
asserts that any non-planar graph must contain a $K_5$ or $K_{3,3}$-minor~\cite{K30,W37}. Thus,
certifying non-planarity is equivalent to producing such a minor, which
can be done in linear time. Can we beat the linear time bound if we knew that
$G$ was ``sufficiently" non-planar?

Assume random access to an adjacency list representation of a bounded degree graph, $G$.
Suppose, for some constant $\eps > 0$, one had to remove $\eps n$ edges from $G$ to make it planar.
Can one find a forbidden ($K_5$ or $K_{3,3}$) minor in $o(n)$ time?
It is natural to ask this question for any property expressible through forbidden minors.
By the famous Robertson-Seymour graph minor theorem~\cite{RS:20}, any graph property $\cP$ that is closed
under taking minors can be expressed by a finite list of forbidden minors.
We desire sublinear time algorithms to find a forbidden minor in any $G$
that requires $\eps n$ edge deletions to make it have $\cP$.

This problem was first posed by Benjamini-Schramm-Shapira~\cite{BSS08}  in the context
of property testing on bounded degree graphs. We follow the model of property testing
on bounded-degree graphs as defined by Goldreich-Ron~\cite{GR02}. Fix a degree bound $d$.
Consider $G = (V,E)$, where $V = [n]$, and $G$ is represented by an adjacency list.
We have random access to the list through \emph{neighbor queries}.
There is an oracle that, given $v \in V$ and $i \in [d]$,
returns the $i$th neighbor of $v$ (if no neighbor exists, it returns $\bot$).

Given any property $\cP$ of graphs with degree bound $d$, the distance of $G$ to $\cP$
is defined to be the minimum number of edge additions/removals required to make $G$ have $\cP$,
divided by $dn$. This ensures that the distance
is in $[0,1]$. We say that $G$ is $\eps$-far from $\cP$ if the distance to $\cP$
is more than $\eps$.

A property tester for $\cP$ is a randomized procedure takes as input (query access to) $G$ and a proximity parameter $\eps > 0$.
If $G \in \cP$, the tester must accept with probability at least $2/3$. If $G$ is $\eps$-far from $\cP$,
the tester must reject with probability at least $2/3$. A one-sided tester must accept $G \in \cP$
with probability $1$, and thus provide a certificate of rejection.

We are interested in property $\cP$ expressible through \emph{forbidden minors}.
Fix a finite graph $H$. The property $\cP_H$ of \emph{$H$-minor freeness} is
the set of graphs that do not contain $H$ as a minor. Observe that one-sided testers for $\cP_H$
have a special significance since they must produce an $H$-minor whenever they reject.
One can cast one-sided property testers as sublinear time procedures that find forbidden minors.
Our main theorem follows.

\begin{theorem} \label{thm:main} Fix a finite graph $H$ with $|V(H)| = r$ and arbitrarily small $\delta > 0$.
Let $\cP_H$ be the property
of $H$-minor freeness. There is a randomized algorithm that takes as input (oracle access to) a graph $G$
with maximum degree $d$, and a parameter $\eps > 0$. Its running time is 
%$O(\eps^{-\exp(1/\delta)} + \eps^{-2}d n^{1/2 + \delta})$. If $G$ is $\eps$-far from $\cP_H$,
$\runtime$. If $G$ is $\eps$-far from $\cP_H$,
then, with probability $> 2/3$, the algorithm outputs an $H$-minor in $G$.

% 
% Suppose $G$ is a graph with maximum degree $d$ and has $n$ vertices.
% Furthermore, $G$ is $\eps$-far from $\cP_H$. For any constant $\delta > 0$, there
% is a randomized algorithm with running time 
% that, with probability $> 2/3$, outputs an $H$-minor in $G$.

Equivalently, there exists a one-sided property tester for $\cP_H$ with the above running time. 
\end{theorem}

The graph-minor theorem of Robertson and Seymour~\cite{RS:20} asserts the following.
Consider any property $\cQ$ that is closed under taking minors. There is a finite list
$\bH$ of graphs such that $G \in \cQ$ iff $G$ is $H$-minor free for all $H \in \bH$.
If $G$ is $\eps$-far from $\cQ$, then $G$ is $\Omega(\eps)$-far from $\cP_H$ for some $H \in \bH$.
Thus, a direct corollary of \Thm{main} is the following.

\begin{corollary} \label{cor:main} Let $\cQ$ be any minor-closed property of graphs
with degree bound $d$. For any $\delta > 0$, there is a one-sided property tester for $\cQ$
with running time $\runtimenor$.
\end{corollary}

In the following discussion, we suppress dependences on $\eps$ and $n^\delta$ by $O^*(\cdot)$
(where $\delta > 0$ is arbitrarily small).
Previously, the only graphs $H$ for which an analogue of \Thm{main} was known are the following:
$O^*(1)$ time for $H$ being a forest, $O^*(\sqrt{n})$ for $H$ being a cycle~\cite{C14},
and $O^*(n^{2/3})$ for $H$ being $K_{2,k}$, the $(k\times 2)$-grid, and the $k$-circus~\cite{FLVW:17}.
No sublinear time bound was known for planarity.

\Cor{main} implies that properties such as planarity, series-parallel graphs, 
embeddability in bounded genus surfaces, and bounded treewidth are all one-sided testable
in $O^*(\sqrt{n})$ time.

We note a particularly pleasing application of \Thm{main}. Suppose bounded degree $G$ has more than
$(3+\eps)n$ edges. Then it is guaranteed to be $\eps$-far from being planar,
and thus, there is an algorithm to find a forbidden minor in $G$ in $O^*(\sqrt{n})$ time.
Since all minor-closed properties have constant average degree bounds, analogous statements
can be made for all such properties.

\subsection{Related work} \label{sec:related}

Graph minor theory is a deep topic, and we refer the reader to Chapter 12 of Diestel's book~\cite{R-book}
and Lov\'{a}sz' survey~\cite{L06}.
For our purposes, we use as a black-box polynomial time algorithms that find fixed minors
in a graph. A result of Kawarabayashi-Kobayashi-Reed provides an $O(n^2)$ time algorithm~\cite{KKR:12}.

Property testing on graphs is an immensely rich area of study, and we refer
the reader to Goldreich's recent textbook for more details~\cite{G17-book}.
There is a significant difference between the theory of property testing for dense graphs
and that of bounded-degree graphs. For the former, there is a complete characterization
of properties (one-sided, non-adaptive) testable in query complexity independent
of graph size. There is a deep connection between property testing and the Szemeredi
regularity lemma~\cite{AFNS06}. Property testing for bounded degree graphs is much less understood.
This study was initiated by Goldreich-Ron, and the first results focused on
connectivity properties~\cite{GR02}. Czumaj-Sohler-Shapira proved that hereditary properties
of non-expanding graphs are testable~\cite{CSS09}.
A breakthrough result of Benjamini-Schramm-Shapira (henceforth BSS) proved that all minor-closed (more generally,
hyperfinite) properties are two-sided testable in constant time. The dependence on $\eps$
was subsequently improved by Hassidim et al, using the concept of local partitioning oracles~\cite{HKNO}.
A result of Levi-Ron \cite{LR15} significantly simplified and improved this analysis, to get
a final query complexity quasi-polynomial in $1/\eps$. Indeed, it is a major open question
to get polynomial dependence on $1/\eps$ for two-sided planarity testers. Towards this goal,
Ito and Yoshida give such a bound for testing outerplanarity \cite{YI:15}, or Edelman et al generalize
for bounded treewidth graphs~\cite{EHNO11}. 
%     Later,  \cite{EHNO11} 
% extended this result and gave a $2$-sided tester (with query complexity $poly(1/\eps)$) for deciding 
% whether a given graph has bounded treewidth or is far from the class of bounded treewidth graphs. In
% particular, the treewidth of an outerplanar graph is at most $2$ and so thus \cite{EHNO11} implies
% a tester for this property with query complexity $poly(1/\eps)$.

In contrast to dense graph testing, there is a significant jump in complexity
for one-sided testers. 
BSS first raised the question of one-sided testers for minor-closed properties (especially planarity)
and conjectured that the bound is $O(\sqrt{n})$.
Czumaj et al ~\cite{C14} made the first step 
by giving an $\otilde(\sqrt{n})$ one-sided tester for the property of being $C_k$-minor free \cite{C14}. For $k=3$, this
is precisely the class of forests. This tester is obtained by a reduction
to a much older result
of Goldreich-Ron for one-sided bipartiteness testing for bounded degree graphs~\cite{GR99}.
(The results in Czumaj et al are obtained by black-box applications of this result.)
Czumaj et al adapt the one-sided $\Omega(\sqrt{n})$ lower bound for bipartiteness
and show an $\Omega(\sqrt{n})$ lower bound for one-sided testers for $H$-minor freeness
when $H$ has a cycle~\cite{C14}. This is complemented with a constant time tester for $H$-minor freeness
when $H$ is a forest. 

Recently, Fichtenberger-Levi-Vasudev-W\"{o}tzel give an $\otilde(n^{2/3})$ tester for $H$-minor freeness
when $H$ is one of the following graphs: $K_{2,k}$, the $(k\times 2)$-grid or the $k$-circus 
graph (a wheel where spokes have two edges) \cite{FLVW:17}. This subsumes the properties of outerplanarity
and cactus graphs. This result uses a different, more combinatorial (as opposed to random walk based) approach than Czumaj et al. 

The use of random walks in property testing was pioneered by Goldreich-Ron~\cite{GR99}
and was then (naturally) used in testing expansion properties and clustering structure~\cite{GR00,CzumajS:10,KS,NS,KalePS:13,CPS15}.
Our approach is inspired by the Goldreich-Ron analysis, and we
discuss more in the next section.
A number of previous results have used random walks for routing in expanders~\cite{BroderFU99,KR:96}. 
We use techniques from Kale-Seshadhri-Peres to analyze random walks on projected Markov Chains~\cite{KalePS:13}.
We also employ the local partitioning methods of Spielman-Teng~\cite{ST12},
which is in turn derived from the Lov\'{a}sz-Simonovits analysis technique~\cite{LS:90}.

% All existing one-sided testers heavily use the specific properties of $H$,
% and it is not clear how they would generalize 

\section{Main Ideas} \label{sec:ideas}

We give an overview of the proof strategy and discuss the various moving parts of the proof.
For convenience, assume that $G$ is a $d$-regular graph.
It is instructive to understand the method of Goldreich-Ron (henceforth GR) for one-side bipartiteness testing~\cite{GR99}.
The basic idea to perform $O(\sqrt{n})$ random walks of $\poly(\log n)$ length from a uar vertex $s$.
An odd cycle is discovered when two walks end at the same vertex $v$, through path of differing parity (of length).

The GR analysis first considers the case when $G$ is an expander
(and $\eps$-far from bipartite).
In this case, the walks from $s$ reach the stationary distribution. 
One can use a standard collision argument to show that $O(\sqrt{n})$ suffice
to hit the same vertex $v$ twice, with different parity paths.
The deep insight is that any graph $G$ can be decomposed into pieces where the algorithm works,
and each piece $P$ has a small cut to $\overline{P}$. 
This has connections with decomposing a graph into expander-like pieces~\cite{TR05, GharanT12}. Famously, the Arora-Barak-Steurer algorithm~\cite{AroraBS:15} for unique games basically
proves such a statement. We note that GR does \emph{not} decompose into expanders,
but rather into pieces where the expander analysis goes through. So, one might hope to analyze the algorithm by its behavior
on each component. Unfortunately, the algorithm cannot produce the decomposition;
it can only walk in $G$ and hope that performing random walks in $G$ suffice
to simulate the procedure within $P$.
This is extremely challenging, and is precisely what GR achieve (this is
the bulk of the analysis). The main lemma produces a decomposition into such pieces,
such that for each piece $P$, there exists $s \in P$ wherein short random
walks (in $G$) from $s$ reach all vertices in $P$ with sufficient probability.
One can think of this a simulation argument: we would like to simulate the random walk
algorithm running only on $P$, through random walks in $G$.

\textbf{The challenge of general minors:} With planarity in mind, let us focus on finding $K_5$ minors.
It is highly unlikely that random walks from a single vertex will find a such a minor.
Intuitively, we would need to find 5 different vertices, launch random walks from all of them
and hope these walks will produce a minor. Thus, we would need to simulate a much more complex
procedure than the (odd) cycle finder of GR. Most significantly, we need to understand
the random walks behavior from multiple sources within $P$ simultaneously. The GR analysis
actually constructs the pieces $P$ by a local partitioning looking at the random walk distribution from a single vertex.
There is no guarantee on random walk behavior from other vertices in $P$.

There is a more significant challenge from arbitrary minors. The simulation does not say anything about the specific structure
of the paths generated. It only deals with the probability of reaching $v$ from $s$ by a random walk in $G$
when $v$ and $s$ are in the same piece. For bipartiteness, as long as we find two paths of differing
parity, we are done. They may intersect each other arbitrarily.
For finding a $K_5$ minor, the actual intersection behavior. We would need paths between all pairs
of 5 seed vertices to be ``disjoint enough" to give a $K_5$ minor. This appears extremely difficult
using the GR analysis. Even if we did understand the random walk behavior (in $G$) from all vertices in $P$,
we have little control over their behavior when they leave $P$. (Based on the parameters, the walks leave
$P$ with high probability.) They may intersect arbitrarily, and thus destroy any minor structure.

\subsection{When do random walks find minors?} \label{sec:mixing}

Inspired by GR, let us start with an algorithm to find a $K_5$ minor in an expander $G$.
(Variants of these ideas were present in a result of Kleinberg-Rubinfeld that 
expanders contain an $H$-minor for any $H$ with $n/\poly(\log n)$ edges~\cite{KR:96}.)
Let $\ell$ denote the mixing time.
Pick u.a.r. a vertex $s$, and launch $5$ random walks each of length $\ell$ to reach $v_1, v_2, \ldots, v_5$.
From each $v_i$, launch $\sqrt{n}$ random walks each of length $\ell$. With high probability,
a walk from $v_i$ and a walk from $v_j$ will ``collide" (end at the same vertex). We can
collect these collisions to get paths between all $v_i, v_j$, and one can, with some effort, show
that these form a $K_5$-minor.

Our main insight is to show that this algorithm, with minor
modifications, works even when random walks have extremely slow mixing properties. 
When the random walks mix even more slowly than the requisite bound, we can essentially
perform local partitioning to pull out very small ($n^\delta$ for arbitrarily small $\delta > 0$)
pieces that have low conductance cuts. We can simply query all edges in this piece and run
a planarity test.

There is a parameter $\delta > 0$ that can be set to an arbitrarily small constant.
Let us set the random walk length $\ell$ to $n^\delta$, and let ${\bf p}_{s,\ell}$ 
be the random walk distribution after $\ell$ steps from $s$. Our proof splits into two
cases, where $\alpha = c\delta$ for explicit constant $c > 1$:
\begin{asparaitem}
    \item Case 1 (the leaky case): For at least $\eps n$ vertices $s$, $\|{\bf p}_{s,\ell}\|^2_2 \leq 1/n^{\alpha}$.
    \item Case 2 (the trapped case): For at least $(1-\eps)n$ vertices $s$,$\|{\bf p}_{s,\ell}\|^2_2 > 1/n^{\alpha}$.
\end{asparaitem}
In the leaky case, random walks are hardly mixing by any standard of convergence.
We are merely requiring that a random walk of length $n^\delta$ (roughly speaking)
spreads to a set of size $n^{c\delta}$. 

\medskip
We prove that, in the leaky case, the procedure described
in the first paragraph succeeds in finding a $K_5$ with high probability.
We give an outline of this proof strategy. 
% 
% 
% , under some unreasonable assumptionsb
% on ${\bf p}_{v,\ell}$. Suppose $\forall v \in V$:
% 
% \begin{asparaitem}
%     \item The distribution ${\bf p}_{v,\ell}$ is uniform on some subset of $V$. We will
%     refer to such distributions as \emph{subset uniform}. Thus, in the leaky case, this distribution
%     is uniform on some set of size at least $n^{10\delta}$.
%     \item ${\bf p}_{v,\ell/2} = {\bf p}_{v,\ell}$
% \end{asparaitem}

Let us assume that ${\bf p}_{v,\ell/2} = {\bf p}_{v,\ell}$ (so $\ell$-length walks have ``stabilized"). 
Let us make a slight modification to the algorithm. We pick $v_1, \ldots, v_5$ as before,
with $\ell$-length random walks from $s$. We will perform $O(\sqrt{n})$ $\ell/2$ length random walks
from each $v_i$ to produce the $K_5$ minor. By symmetry of the random walks, the probability that a single walk from $v_i$ and one
from $v_j$ collide (to produce a path) is exactly ${\bf p}_{v_i,\ell/2} \cdot {\bf p}_{v_j,\ell/2}$.
Thus, we would like these dot products to be large. 
By the symmetry of the random walk, the probability of an $\ell$-length random walk
starting from $s$ and ending at $v$ is ${\bf p}_{s,\ell/2}\cdot {\bf p}_{v,\ell/2}$.
In other words, the entries of ${\bf p}_{s,\ell}$ are precisely these dot products,
and $\|{\bf p}_{s,\ell}\|^2_2 = \sum_{v \in V} ({\bf p}_{s,\ell/2}\cdot {\bf p}_{v,\ell/2})^2$
$ = \EX_{v \sim {\bf p}_{s,\ell/2}}[{\bf p}_{s,\ell/2}\cdot {\bf p}_{v,\ell/2}]$.
Since ${\bf p}_{s,\ell/2} = {\bf p}_{s,\ell}$, we rewrite to get 
${\bf p}_{s,\ell/2}\cdot {\bf p}_{s,\ell/2} = \EX_{v \sim {\bf p}_{s,\ell/2}}[{\bf p}_{s,\ell/2}\cdot {\bf p}_{v,\ell/2}]$.

Think of the dot products as correlations between distributions.
We are saying that the average correlation (over some distribution on vertices) of ${\bf p}_{v,\ell/2}$ 
with ${\bf p}_{s,\ell/2}$ is exactly the self-correlation of ${\bf p}_{s,\ell/2}$. 
If the distributions by and large had low $\ell_2$-norm (as in the leaky case), we might
hope that these distributions are reasonably correlated with each other. 
Indeed, this is what we prove. Under some conditions, we show that 
$\EX_{v_i, v_j \sim {\bf p}_{s,\ell/2}}[{\bf p}_{v_i,\ell/2} \cdot {\bf p}_{v_j,\ell/2}]$ can
be lower bounded, and ${\bf p}_{v_i,\ell/2}$ is exactly the distribution the algorithm
picks the $v_i$'s from. This is evidence that $\ell/2$-length random walks will 
connect the $v_i$'s through collisions.

There are four difficulties in increasing order of worry. 
\medskip

\begin{asparaenum}
    \item We only have a lower bound of the average ${\bf p}_{v_i,\ell/2} \cdot {\bf p}_{v_j,\ell/2}$. 
    We would need bounds for all (or most) pairs to produce a minor.
    \item ${\bf p}_{v,\ell}$ might be very different from ${\bf p}_{v,\ell/2}$.
    \item The expected number of collisions between walks from $v_i$ and $v_j$
is controlled by the dot product above, but the variance (which really controls the probability
of getting a collision) can be large.
There are instances where the dot product is high, but the collision probability is extremely low.
    \item There is no guarantee that these paths will produce a minor since we do not
have any obvious constraints on the intermediate vertices in the path.
\end{asparaenum}

\medskip

The first problem is surmounted by a technical trick.
It turns out to be cleaner to analyze the probability of getting a biclique minor.
So, we perform $50$ random walks from $s$ to get sets $A = \{a_1, a_2,\ldots,a_{25}\}$
and an analogous $B$. We launch $\ell/2$-length random walks from each vertex in $A\cup B$.
The average lower bound on the dot product suffices to get a lower bound on the probability
of getting a $K_{25,25}$-minor, which contains a $K_5$-minor.

For the second problem, it turns out that the weaker 
bound of $\|{\bf p}_{v,\ell}\|_2 = \Omega(n^{-\delta} \|{\bf p}_{v,\ell/2}\|_2)$ suffices.
We could try to search for some value of $\ell$ where this happens. If there was no (small) value
of $\ell$ where this bound held, then it suggest that $\|{\bf p}_{v,n^{\delta}}\|_2$ is extremely
small (say $\Theta(1/n)$). This kind of reasoning is detailed more in the next subsection.

The third problem requires bounds on the variance, or higher norms, of ${\bf p}_{v,\ell/2}$.
Unfortunately, there appears be no handle on these. At a high level, our idea is to truncate ${\bf p}_{v,\ell/2}$
by ignoring large entries. This truncated vector is not a probability vector any more, but
we can hope to redo the analysis for such vectors.

Now for the fourth problem. Naturally, if the vertices $v_1, \ldots, v_5$
are close to each other, we do not expect to get a minor by connecting them. Suppose they were sufficiently ``spread out",
One could hope that the paths connecting the $v_i, v_j$ pairs would only intersect ``near" the $v_i$.
The portion of the paths nears the $v_i$s could be contracted to get a $K_5$-minor. 
We can roughly quantify how far the $v_i$s will be by the variance of ${\bf p}_{v,\ell/2}$.
Thus, the third and fourth problem are coupled.

% 
% 
% 
% 
% If one could assume that these distributions were uniform in their support, the second
% and third problems go away. All in all, under the unreasonable assumptions that the distributions are uniform
% in their support (call these \emph{support uniform}) and ${\bf p}_{s,\ell/2} = {\bf p}_{s,\ell}$,
% the leaky case can be handled.

\subsection{$R$-returning walks} \label{sec:return}

The main technical contribution of our work is in defining $R$-returning walks.
These are walks that periodically return to a given set $R$ of vertices. A careful analysis
of these walks provides to tools to handle the various problems discussed above.

Fix $\ell$ as before.
Formally, an $R$-returning walk of length $j\ell$ (for $j \in \NN$)
is a walk that encounters $R$ at every $i\ell$ step $\forall i \in [j]$.
While random walk distributions can have poor variance,
we can carefully choose $R$ to ensure that the distribution of $R$-returning
walks is well-behaved. We will quantify this as approximate ``support uniformity"
(being approximatedly uniform on the support).

In the leaky case, there is some (large) set $R$, such that $\forall s \in R$, $\|{\bf p}_{s,\ell/2}\|^2_2 \leq 1/n^\alpha$.
Let ${\bf p}_{[R],s,\ell}$ be the random walk distribution
restricted to $R$. Suppose for some $s \in R$, $\|{\bf p}_{[R],s,\ell}\|^2_2 \geq 1/n^{\alpha+\delta}$.
Observe that each entry in ${\bf p}_{[R],s,\ell}$ is ${\bf p}_{s,\ell/2} \cdot {\bf p}_{v,\ell/2}$,
for $s,v \in R$. By Cauchy-Schwartz, this is at most $1/n^\alpha$. For any distribution ${\bf v}$,
the condition $\|{\bf v}\|^2_2 = \|{\bf v}\|_\infty$ is equivalent to support uniformity.
Thus, ${\bf p}_{[R],s,\ell}$ is approximately support uniform, up to $n^\delta$ deviations.
The math discussed in the previous section goes through for any such $s$. In other words,
if the random walk algorithm started from $s$, it succeeds in finding a $K_5$ minor.

Suppose only a negligible fraction of vertices satisfied this condition, and so our algorithm
would not actually find such a vertex. Let us remove all these vertices from $R$ (abusing
notation, let $R$ be the resulting set). Now, $\forall s \in R$, $\|{\bf p}_{[R],s,\ell}\|^2_2 \leq 1/n^{\alpha+\delta}$.
So, the bound on the $l_2$-norm has fallen by an $n^\delta$ factor.
What does ${\bf p}_{[R],s,\ell} \cdot {\bf p}_{[R],v,\ell}$ signify? This is the probability
of a $2\ell$-length random walk starting from $s$, ending at $v$, and encountering $R$
at the $\ell$th step. This is an $R$-returning walk of length $2\ell$. Let ${\bf q}_{[R],s,2\ell}$
denote the vector of $R$-returning walk probabilities. Suppose for some $s$, $\|{\bf q}_{[R],s,2\ell}\|^2_2 \geq 1/n^{\alpha+2\delta}$.
By Cauchy-Schwartz, $\|{\bf q}_{[R],s,2\ell}\|_\infty \leq 1/n^{\alpha+\delta}$, implying
that ${\bf q}_{[R],s,2\ell}$ is approximately support uniform. Again, the math of the previous section
goes through for such an $s$.

We remove all vertices that have this property, and end up with $R$ such
that $\forall s \in R$, $\|{\bf q}_{[R],s,2\ell}\|^2_2 \leq 1/n^{\alpha+2\delta}$. 
Observe that ${\bf q}_{[R],s,2\ell} \cdot {\bf q}_{[R],v,2\ell}$ is a probability of a $4\ell$ $R$-returning
walk. We then iterate this argument.

In general, this argument goes through phases. In the $i$th phase, we find $s \in R$
that satisfy $\|{\bf q}_{[R],s,2^i\ell}\|^2_2 \geq 1/n^{\alpha+i\delta}$. We show that
the random walk procedure of the previous section (with some modifications) finds a $K_5$-minor
starting from such vertices. We remove all such vertices from $R$, increment $i$ and continue
the argument. The vertices removed at the $i$th phase are called the \emph{$i$th stratum},
and we refer to this entire process as stratification. Intuitively, for vertices in the $i$th stratum,
the $R$-returning (for the setting of $R$ at that phase) walk probabilities roughly form a uniform distribution
of support $n^{\alpha + i\delta}$. Thus, for vertices in higher strata, the random walks are spreading
to larger sets.

There is a major problem. The ${\bf q}$ vectors are \emph{not} distributions,
and the vast majority of walks are not $R$-returning. Indeed, the reduction in norm
as we increase strata might simply be an artifact of the lower probability of a longer
$R$-returning walk (note that the walks lengths are increasing exponentially in the phase number).
We prove a spectral lemma asserting that this is not the case. As long as $R$ is sufficiently large,
the probabilities of $R$-returning walks are sufficiently high. Unfortunately, these 
probabilities (must) decrease exponentially in the number of returns. In the $i$th phase,
the walk length is $2^i\ell$ and it must return to $R$ $2^i$ times. Here is where
the $n^\delta$ decay in $l_2$-norm condition saves us. After $1/\delta$ phases, the 
$\|{\bf q}_{[R],s,2^i\ell}\|^2_2$ is basically $1/n$. The spectral lemma tells us that
if $R$ is still large, the probability that a $2^{1/\delta}\ell$ length walk is $R$-returning
is sufficiently large. Thus, the norm cannot decrease, and almost all vertices
end up in the very next stratum. 
If $R$ was small, then there is an earlier stratum containing $\Omega(\delta \eps n)$ vertices.
Regardless of the case, there exists a $i \leq 1/\delta + O(1)$ such that the $i$th
stratum contains $\Omega(\delta \eps n)$ vertices. For all these vertices, the random walk
algorithm to find minors succeeds with non-trivial probability.

\subsection{The trapped case: local partitioning to the rescue} \label{sec:local}

In this case, for almost all vertices $\|{\bf p}_{s,\ell}\|^2_2 \geq 1/n^\alpha$.
The proofs of the (contrapositive of the) Cheeger inequality basically imply the existence of a set of low condutance cut $P_s$ ``around" $s$. 
By local partitioning methods such as those of Spielman-Teng and Anderson-Chung-Lang~\cite{ST12, ACL06},
we can actually find $P_s$ in roughly $n^\alpha$ time. We expect our graph to basically
decompose into $O(n^\alpha)$ sized components with few edges between them. Our algorithm
can simply find these pieces $P_s$ and run a planarity test on them.
We refer to this as the \emph{local search} procedure.

While the intuition is correct, the analysis is difficult. 
The main problem is that actual partitioning of the graph (into small components
connected by low conductance cuts) is fundamentally iterative. It starts by finding a low conductance set $P_{s_1}$,
then finding a low conductance set $P_{s_2}$ in $\overline{P_{s_1}}$, then
$P_{s_3}$ in $\overline{P_{s_1} \cup P_{s_2}}$, and so on. In general, this requires
conditions on the random walk behavior inside $\overline{\bigcup_{j < i} P_{s_j}}$. On the other hand,
our algorithm and the trapped case condition only refer to random walk behavior in all of $G$.
Furthermore, $\overline{\bigcup_{j < i} P_{s_j}}$ can be as small as $\Theta(\eps n)$, and so 
we do expect the random walk behavior to be quite different.

The GR bipartiteness analysis surmounts this problem and performs such a decomposition, but their parameters do not work
for us. Starting from a source vertex $s$, their analysis discovers $P_s$ such that probabilities
of reaching any vertex in $P_s$ (from $s$) is roughly uniform \emph{and} smaller than $1/\sqrt{n}$.
On the other hand, we would like to discover all of $P_s$ in $n^{O(\delta)}$ time so that we can run
a full planarity test.

We employ a collection of tools, and use the methods of Kale-Peres-Seshadhri to analyze
``projected" Markov Chains~\cite{KalePS:13}. In the analysis above, we have some set $S$ ($\overline{\bigcup_{j < i} P_{s_j}}$)
and want to find a low conductance set $P$ completely contained in $S$. Moreover, we wish to discover
$P$ using random walks in $G$. We construct a Markov chain, $M_S$, with vertex set $S$, and include new
transitions that correspond to walks in $G$ whose intermediate vertices are not in $S$. Each such transition
has an associated ``cost," corresponding to the actual length in $G$. (GR also have a similar idea, although
their Markov chain introduces extra vertices to track the length of the walk in $G$. This makes the analysis
somewhat unwieldy, since low conductance cuts in $M_S$ may include these extra vertices.)

Using bounds on the return time of random walks, we have relationships between the average length
of a walk in $G$ whose endpoints are in $S$ and the corresponding length when ``projected" to $M_S$. 
On average, an $\ell$-length walk in $G$ with endpoints in $S$ corresponds to an $\ell |S|/n$-length walk
in $M_S$. Roughly speaking,
we hope that for many vertices $s$, an $\ell|S|/n$-length walk in $M_S$ is trapped
in a set of size $n^\alpha$. 
% Unfortunately, the variance of the walk length can be quite high. But, given our parameters, we can simply use
% a Markov bound to get statements for length from a specific $s$. 

We employ the Lov\'{a}sz-Simonovits curve technique to produce a low conductance cut $P_s$ in $M_S$~\cite{LS:90}.
We can guarantee that all vertices in $P_s$ are reachable with roughly $n^{-\alpha}$ probability from $s$ through $\ell|S|/n$-length random walks in $M_S$.
Using the average length correspondence between walks in $M_S$ to $G$, we can make
a similar statement in $G$ -  albeit with a longer length.
We basically iterate over this entire argument to produce the decomposition into low conductance pieces.

In our analysis, we use the stratification itself to (implicitly) distinguish between the leaky and trapped case.
Stratification peels the graph into $1/\delta + O(1)$ strata. If a vertex $s$ lies in a stratum numbered
at least some fixed constant $b$, we can show that the algorithm finds a $K_r$-minor with $s$ as the starting
vertex. Thus, if at least (say) $n^{1-\delta}$ vertices lie in stratum $b$ or higher, we are done.
If $s$ is in a low strata, we have a lower bound on the random walks norm. This allows for local partitioning
around $s$. 

% \subsection{More on stratification, and putting the pieces together}
% 
% 
% 
% To reiterate, it is convenient to think of a bounding degree graph as an ``onion'' and all
% the strata as the ``peels'' of this onion. That is, every peel should be thought of as a subset
% of vertices. Thus, our spectral lemma essentially says that the first few peelings take away 
% all of the onion -- their union essentially contains all the vertices. Moreover, if we are told
% that if at least $\eps n$ vertices survive the first few peelings (or so to speak, the ``bulk'' 
% left after $\Omega(r^4)$ peelings is large) then the random walk procedure succeeds with high
% probability in finding the minor. \Sec{local} describes what our algorithm does when the ``bulk''
% is small in size (that is it contains $O(\eps n)$ vertices). 
% 

\section{The algorithm}\label{sec:the-algorithm}

We are given a bounded degree graph $G = (V,E)$, with max degree $d$.
We assume that $V = [n]$. We follow the standard adjacency list model of Goldreich-Ron
for (random) access to the graph. This model allows an algorithm to sample u.a.r. vertices and perform \emph{edge queries}. 
Given a pair $(v,i) \in [n] \times [d]$, the output of an edge query is the $i$th neighbor of $v$
according to the adjacency list ordering. If the degree of $v$ is smaller than $i$,
the output is $\bot$.

In the algorithm, the phrase ``random walk" refers to a lazy random walk on $G$. 
Given a current vertex $v$, with probability $1/2$, the walk remains at $v$.
With probability $1/2$, the procedure generates u.a.r. $i \in [d]$.
It performs the edge query for $(v,i)$. If the output is $\bot$, the walk remains at $v$,
otherwise the walk visits the output vertex.
This is a symmetric, ergodic Markov chain with a uniform stationary distribution.

% We define the behavior of the random walks as follows. If the current vertex in the random walk is $v$, then with probability $1/2$ the walk remains at $v$, and with probability $1/2$, we generate a number $i \in [d]$ uniformly at random. Then we perform the edge query $(v, i)$. If the output is $\bot$, the walk remains at $v$. Otherwise, the walk visits the $i^{\textrm{th}}$ neighbor of $v$. Note that the stationary distribution of the random walk is uniform.
% 
% 
% 
% an oracle which orders the $deg(v)$ neighbors of $v$ in some consistent way and if $i \leq deg(v)$, the oracle returns the $i^{th}$ neighbor of $v$ according to the ordering. Otherwise, if $deg(v) < i \leq d$, then the oracle returns $\bot$. 
%  with $n$ vertices and $m$ edges.

Our main procedure \isminorfree$(G,\eps,H)$, tries to find a $H$-minor in $G$. We prove that it succeeds with high
probability if $G$ is $\eps$-far from being $H$-minor free. There are three subroutines:

\medskip

\begin{asparaitem}
    \item \localsearch$(s)$: This procedure perform a small number of short random walks to find
    the piece described in \Sec{local}. This produces a small subgraph of $G$, where an exact $H$-minor finding
    algorithm is used.
    \item \findpath$(u,v,k,i)$: This procedure tries to find a path from $u$ to $v$. The parameter $i$
    decides the length of the walk, and the procedure performs $k$ walks from $u$ and $v$. If any pair
    of these walks collide, this path is output.
    \item \findclique$(s)$: This is the main procedure mostly as described in \Sec{mixing}. It attempts
    to find a sufficiently large biclique minor. First, it generates seed sets $A$ and $B$ by performing
    random walks from $s$. Then, it calls \findpath{} on all pairs in $A \times B$.
\end{asparaitem}

\medskip

We fix a collection of parameters.
\begin{asparaitem}
	\item $\delta$: An arbitrarily small constant.
	\item $r$: The number of vertices in $H$.
	\item $\ell$: The random walk length. This will be $n^{5\delta}$.
% 	\item $\alpha$: a parameter for the decomposition lemma, we set it to $\frac{2 \eps}{5 \log n}$
	\item $\cutoff$: $\cutoff = n^{\frac{-\delta}{exp(2/\delta)}}$. If $\eps < \cutoff$, the algorithm just queries the whole graph.
    \item \RS$(F,H)$: This refers to an exact $H$-minor finding process (in $F$). For concreteness, we use the quadratic time procedure of
    Kawarabayashi-Kobayashi-Reed~\cite{KKR:12}.
\end{asparaitem}

% \Andrew{since we have a different algorithm, need to rewrite this.} Our algorithm explores small subgraphs of $G$ and looks for minors inside these subgraphs. In these cases, we use a super linear deterministic algorithm to explicitly find minors. The algorithm we use for this is the Robertson-Seymour algorithm\Andrew{cite}, and we will refer to this algorithm as \RS $(G, H)$ for some graph $G$ and target minor $H$.

% In this section, we present our algorithm for testing the minor-freeness of graphs. The algorithm performs two different procedures to identify forbidden minors at two different scales. The first, \localsearch explores small (on the order of $n^{o(1)}$ vertices) neighborhoods of the graph and searches the induced subgraph for the forbidden minor. The second procedure, \findclique, uses random walks to identify minors which may not be contained in the neighborhood of a single vertex. Combining both of these procedures, we show that we can find minors in all graphs which are far from minor-free.

\medskip
\noindent
\fbox{
	\begin{minipage}{0.9\textwidth}
		{\isminorfree$(G, \eps, H)$}
		
		\smallskip
		\begin{compactenum}
		\item If $\eps < \cutoff$, query all of $G$, and output \RS$(G,H)$
		\item Else \begin{compactenum}
			%\item Let $r := |V(H)|$
			\item Repeat $\valone$ times:
			\begin{compactenum}
				\item Pick uar $s \in V$
				\item Call \localsearch$(s)$ and \findclique$(s)$.
			\end{compactenum}
% 			\item Repeat $\rho_2 = \valtwo$ times:
% 			\begin{compactenum}
% 				\item Pick uar $s \in V$
% 				\item Call \findclique$(s)$. If it rejects, \reject and \exit
% 			\end{compactenum}
% 			\item \accept and \exit
		\end{compactenum}
		\end{compactenum}
\end{minipage}} \\

\noindent
\fbox{
	\begin{minipage}{0.9\textwidth}
		{\localsearch$(s)$}
		
		\smallskip
		\begin{compactenum}
		        \item Initialize set $B = \emptyset$.	
				\item For $h = 1, \ldots, n^{7 \delta r^4}$:
				\begin{compactenum}
					\item Perform $\eps^{-1} n^{30 \delta r^4}$ independent random walks of length $h$ from $s$.
                    Add all destination vertices to $B$.
				\end{compactenum}
	\item Determine $G[B]$, the subgraph induced by $B$.
	\item Run $\RS(G[B], H)$. If it returns an $H$-minor, output that and terminate.
		\end{compactenum}
\end{minipage}} \\

\noindent
\fbox{
	\begin{minipage}{0.9\textwidth}
		{\tt \findclique$(s)$}
		
		\smallskip
		\begin{compactenum}
			\item For $i = 5r^4, \ldots, 1/\delta + 4$:
			\begin{compactenum}
				\item Perform $2r^2$ independent random walks of length $2^{i+1} \ell$ from $s$. Let the destinations of the first $r^2$
                walks be multiset $A$, and the destinations of the remaining walks be $B$.
				\item For each $a \in A$, $b \in B$:
				\begin{compactenum}
					\item Run \findpath$(a,b,n^{\findpathexp},i)$
				\end{compactenum}
				\item If all calls to \findpath{} return a path, then let the collection of paths be the subgraph $F$.
                Run $\RS(F,H)$. If it returns an $H$-minor, output that and terminate.
			\end{compactenum}
		\end{compactenum}
\end{minipage}}

\noindent
\fbox{
	\begin{minipage}{0.9\textwidth}
		{\tt \findpath$(u,v,k,i)$}
		
		\smallskip
		\begin{compactenum}
			\item Perform $k$ random walks of length $2^i \ell$ from $u$ and $v$.
			\item If a walk from $u$ and $v$ terminate at the same vertex, return these paths. (Otherwise, return nothing.)
		\end{compactenum}
\end{minipage}}

\begin{theorem}\label{thm:main-result} 
If $G$ is $\eps$-far from being $H$-minor free, then \isminorfree$(G,\eps,H)$ finds an $H$-minor of $G$
with probability at least $2/3$. Furthermore, \isminorfree{} has a running time of $\runtime$. 
\end{theorem}

The query complexity is fairly easy to compute. The total queries made in the \localsearch{}
calls is $d n^{O(\delta r^4)}$. The main work happens in the calls of \findpath, within \findclique{}.
Observe that $k$ is set to $n^{\findpathexp}$, where $i \leq 1/\delta + 4$. This leads to the $\sqrt{n}$
in the final complexity. 
(In general, a setting of $\delta < 1/\log (\eps^{-1}\log\log n)$ suffices for an $n^{1/2 + o(1)}$ running time.)
% choosing $\delta = 1/\log{\log \log n}$, we get the $n^{1/2 + o(1)}$ result claimed in the paper.

\textbf{Outline:} There are a number of moving parts in the proof, which we relegate to their own subsections.
We first develop the notion of $R$-returning walks and the stratification process,
given in \Sec{returning}. In \Sec{path}, we use these techniques to prove that \findclique{}
discovers a sufficiently large biclique-minor in the leaky case. In \Sec{partition}, we prove a local partitioning lemma
that will be used to handle the trapped case. Finally, in \Sec{final},
we put the tools together to complete the proof of \Thm{main-result}.

\section{Returning walks and stratification} \label{sec:returning}

We introduce the concept of $R$-returning random walks for any $R \subseteq V$. 
These definitions are with respect to a fixed length $\ell$.
% These are random walks on $G$ that return to the set $R$ every $\ell$ steps. 

\begin{definition} \label{def:trun} For any set of vertices $R$, $s \in R$, $u \in R$, and $i \in \NN$,
we define the \emph{$R$-returning probability} as follows. We denote by $\trun{s}{u}{R}{i}$
the probability that a $2^i\ell$-length random walk from $s$ ends at $u$, and
encounters a vertex in $S$ at every $j\ell^{\textrm{th}}$ step, for all $1 \leq j \leq 2^i$.
The $R$-returning probability vector, denoted by $\trunvec{s}{R}{i}$, is 
the $|R|$-dimensional vector of returning probabilities.
\end{definition}

\begin{proposition} \label{prop:trun} $\trun{s}{u}{R}{i+1} = \trunvec{s}{R}{i} \cdot \trunvec{u}{R}{i}$
\end{proposition}

\begin{proof} We use the symmetry of (returning) random walks in $G$.
$$\trun{s}{u}{R}{i+1} = \sum_{w \in S} \trun{s}{w}{R}{i} \trun{w}{u}{R}{i}
= \sum_{w \in R} \trun{s}{w}{R}{u} \trun{u}{w}{R}{i} = \trunvec{s}{R}{i} \cdot \trunvec{u}{R}{i}$$
\end{proof}

Let $M$ be the transition matrix of the lazy random walk on $G$. Let $\mathbb{P}_R$
be the $n \times |R|$ matrix on $R$, where each column is the unit vector for some $s \in R$.
For any set $U$, we use $\bone_U$ for the indicator vector on $U$. If no subscript is given,
it is the all ones vector, for the appropriate dimension.

\begin{proposition} \label{prop:mat} $\trunvec{s}{R}{i} = (\mathbb{P}^T_R M^{\ell} \mathbb{P}_R)^{2^i}\bone_s$
\end{proposition}

\noindent Now for a critical lemma. We can lower bound the total probability of an $R$-returning random walk. 
If $R$ contains at least a $\beta$-fraction of vertices, the average $R$-returning walk probability,
for $t$ returns, is at least $\beta^t$.

\begin{lemma} \label{lem:trun} $|R|^{-1} \sum_{s \in R} \|\trunvec{s}{R}{i}\|_1 \geq (|R|/n)^{2^i}$
\end{lemma}

\begin{proof} We will express $\sum_{s \in R} \|\trunvec{s}{R}{i}\|_1 = \bone^T(\mathbb{P}^T_R M^{\ell} \mathbb{P}_R)^{2^i}\bone$.
Let us first prove the lemma for $i=0$. Observe that $\sum_{s \in R} \|\trunvec{s}{R}{0}\|_1 = \bone^T_R M^\ell \bone_R
= ((M^T)^{\ell/2} \bone_R)^T (M^{\ell/2} \bone_R) = \|M^{\ell/2} \bone_R\|^2_2$.
Since $M^{\ell/2}$ is a stochastic matrix, $\|M^{\ell/2} \bone_R\|_1 = \|\bone_R\|_1 = |R|$.
By a standard norm inequality, $\|M^{\ell/2}\bone_R\|^2_2 \geq \|M^{\ell/2}\bone_R\|^2_1/n = |R|^2/n$.
This completes the proof for $i=0$.

Let $N = \mathbb{P}^T_R M^\ell \mathbb{P}_R$, which is a symmetric matrix.
We have just proven that $\bone^TN \bone \geq |R|^2/n$. Let the eigenvalues
of $N$ be $1 \geq \lambda_1 \geq \lambda_2 \ldots \lambda_{|R|}$, with corresponding
eigenvectors $\bf{u}_1, \bf{u}_2, \ldots, \bf{u}_s$. 
We can express $\bone = \sum_{k \leq |R|} \alpha_k \bf{u}_k$, where $\sum_k \alpha^2_k = |R|$.
Observe that $N^{2^i} \bone = \sum_{k \leq |R|} \alpha_k \lambda^{2^i}_k \bf{u}_k$

Let $\mu_k = \alpha^2_k/\sum_j \alpha^2_j$, noting that $\sum_k \mu_k = 1$. We apply Jensen's inequality below.
\begin{eqnarray*}
    \frac{\bone^T N^{2^i} \bone}{|R|} = \frac{\sum_k \alpha^2_k \lambda^{2^i}_k}{\sum_j \alpha^2_j}
= \sum_k \mu_k \lambda^{2^i}_k \geq (\sum_k \mu_k \lambda_k)^{2^i}
\end{eqnarray*}

For $i=0$, we already proved that $\bone^T N \bone/|R| = \sum_k \mu_k \lambda_k \geq |R|/n$. We plug this bound
to complete the proof for general $i$.
% 
% and $\bone^T N^{2^i} \bone = \sum_{k \leq |R|} \alpha^2_k \lambda^{2^i}_k$.
% 
% 
% 		\noindent Note that this also means for $i = 0$ case, we can express the LHS as 
% 		$$\frac{\bone^TN \bone}{|R|} = \frac{\bone^TN \bone}{\sum \alpha_k^2} = \frac{1}{\sum \alpha_k^2} \cdot \sum \alpha_k^2 \lambda_k.$$ 
% 		This can be expressed as $\sum \mu_k \lambda_k$ where $\mu_k = \frac{\alpha_k^2}{\sum \alpha_k^2}$.
% 
% \noindent Now, we deal with general $i$.
% \noindent The first inequality above follows using Jensen's. The last inequality uses the warm up we did above.
% 
\end{proof}

% Let us make a few comments on the foregoing lemma. It crucially points out that  
% large enough sets $R$ -- which contain an $\beta = \Omega(1)$ fraction of vertices -- have the 
% property that a returning walk with $t$ returns started from a random vertex in $R$ has ``large''
% $\ell_1$ mass which decays exponentially in the number of returns: as $\beta^t$. Notice that
% this does not depend on the periodicity $\ell$ of returns.
% 

% In this section, we consider the stratification process of \Def{strata}. We present several results that will be used to analyze \findclique as well as a proof of \Lem{strata}.
\subsection{Stratification} \label{sec:strata}

% A convenient caricature of the stratification process was outlined in \Sec{return}. There, we suggested that it is convenient to imagine that the vertex set of a bounded degree graph  can be likened to an ``onion''. The stratification process essentially peels the onion. Using \Lem{trun}, we will show that a few peelings (at most $1/\delta + 2$) essentially contain all the vertices. The way we setup this process, it will follow that the higher strata/peels contain vertices random walks from which tend to mix faster and the set of vertices in lower strata/peels produce random walks which tend to get stuck. We will be interested in seeing whether the union of vertices in the last few peels (from strata $5r^4$ to strata $1/\delta$) has size $\Omega(\eps n)$ or not. \\
% 

Stratification results in a collection of disjoint sets of vertices denoted by $S_0, S_1, \ldots$ which are called \emph{strata}. 
The corresponding \emph{residue} sets denoted by $R_0, R_1, \ldots$. The zeroth residue $R_0$ is initialized before
stratification and subsequent residues are defined 
by the recurrence $R_i = R_0 \setminus \bigcup_{j < i}S_j$. 
The definitions and claims may seem technical, and the proofs are mostly norm manipulations. But these provide
the tools to analyze our main algorithm.

\begin{definition} \label{def:strata} Suppose $R_i$ has been constructed. A vertex $s \in R_i$
	is placed in $S_i$ if $\|\trunvec{s}{R_i}{i+1}\|^2_2 \geq 1/n^{\delta i}$.
\end{definition}

We have an upper bound for the length of $R_i$-returning walk vectors.

\begin{claim} \label{clm:strata} For all $s \in R_i$ and $1 \leq j \leq i$, $\|\trunvec{s}{R_i}{j}\|^2_2 \leq 1/n^{\delta(j-1)}$.
\end{claim}

\begin{proof} Suppose $\exists j \leq i, \|\trunvec{s}{R_i}{j}\|^2_2 > 1/n^{\delta(j-1)}$.
%By \Prop{mat} and the fact that all eigenvalues of $M$ have absolute value at most $1$,
%$\|\trunvec{s}{R_i}{j}\|^2_2 \geq \|\trunvec{s}{R_i}{i}\|^2_2 > 1/n^{\delta(j-1)}$.
By assumption, $s \in R_i \subseteq R_{j-1}$.  An $R_{i}$-returning walk from $s$ is also
an $R_{j-1}$-returning walk. Thus, every entry of $\trunvec{s}{R_{j-1}}{j}$ 
is at least that of $\trunvec{s}{R_{i}}{j}$. So $\|\trunvec{s}{R_{j-1}}{j}\|^2_2 \geq \|\trunvec{s}{R_i}{j}\|^2_2$
$ > 1/n^{\delta(j-1)}$.
%$\geq \|\trunvec{s}{R_i}{i}\|^2_2 > 1/n^{\delta(j-1)}$.
This implies that $s \in S_{j-1}$ or an earlier stratum, contradicting the assumption that $s \in R_i$.
% Also, $\|\trunvec{s}{R_i}{i+1}\|_1 \geq \eps^{2^{3/\delta}}$.
% Note that $\|\trunvec{s}{R_i}{i+1}\|_1 \leq \|\trunvec{s}{R_{i-1}}{i+1}\|_1$. By definition,
% truncated walks for $i+1$ must be extensions of truncated walks for $i$.
% Thus, $\|\trunvec{s}{R_{i-1}}{i+1}\|_1 \leq \|\trunvec{s}{R_{i-1}}{i}\|_1$.
% Combining, we get that $\|\trunvec{s}{R_{i-1}}{i}\|_1 \geq \eps^{2^{3/\delta}}$.
% 
% Analogously, $\|\trunvec{s}{R_{i-1}}{i}\|^2_2 \geq \|\trunvec{s}{R_i}{i}\|^2_2$, which is greater than $1/n^{\delta(i-1)}$,
% by assumption. This implies that $s \in S_{i-1}$, contradicting the fact that $s \in R_i$.
\end{proof}

We prove an $\ell_\infty$ bound on the returning probability vectors.
Note that we allow $j$ to be $i+1$ in the following bound.

\begin{claim} \label{clm:max-prob} For all $s \in R_i$ and $2 \leq j \leq i+1$, $\|\trunvec{s}{R_i}{j}\|_\infty \leq 1/n^{\delta(j-2)}$.
\end{claim}

\begin{proof} By \Prop{trun}, for any $v \in R_i$, $\trun{s}{v}{R_i}{j} = \trunvec{s}{R_i}{j-1}\cdot \trunvec{v}{R_i}{j-1}$.
Note that $1 \leq j-1 \leq i$. By Cauchy-Schwartz and \Clm{strata}, $\trun{s}{v}{R_i}{j} \leq 1/n^{\delta(j-2)}$.
\end{proof}

As a consequence of these bounds, we are able to bound the amount of probability mass retained by $R_i$-returning walks.

\begin{claim}\label{clm:1-norm}
	For all $s \in S_i$, $||\trunvec{s}{R_i}{i+1}||_1 \geq n ^{-\delta}$. 
\end{claim}
\begin{proof}
	Since $s \in S_i$, $||\trunvec{s}{R_i}{i+1}||_2^2 \geq n^{-i\delta}$, and by \Clm{max-prob}, 
    $||\trunvec{s}{R_i}{i+1}||_\infty \leq n^{-\delta (i - 1)}$. Since,
     $||\trunvec{s}{R_i}{i+1}||^2_2 \leq ||\trunvec{s}{R_i}{i+1}||_1 ||\trunvec{s}{R_i}{i+1}||_\infty$,
     we conclude $||\trunvec{s}{R_i}{i+1}||_1 \geq n^{ - i \delta} n^{\delta(i - 1)} = n^{-\delta}$.
\end{proof}

We prove that most vertices lie in ``early" strata.

\begin{lemma} \label{lem:strata} 
		Suppose $\eps \geq \cutoff$. At most $\eps{n}/{\log n}$ vertices are in $R_{1/\delta + 3}$.
\end{lemma}

\begin{proof} We prove by contradiction. Suppose that $R_{1/\delta+3}$ has at least $\eps n/\log n$ vertices. 
The previous residue, $R_{1/\delta + 2}$, is only bigger and thus $|R_{1/\delta + 2}| \geq \eps n/\log n$ as well. By \Lem{trun},
\begin{equation}
	|R_{1/\delta+2}|^{-1} \sum_{s \in R_{1/\delta+2}} \|\trunvec{s}{R_{1/\delta+2}}{1/\delta+3}\|_1 \geq \left(\frac{\eps}{\log n}\right)^{2^{1/\delta+3}} \textrm{.}
\end{equation}
\noindent By averaging and a standard $l_1$-$l_2$ norm inequality,
\begin{equation}
	||\trunvec{s}{R_{1/\delta+2}}{1/\delta+3}\|^2_2 \geq n^{-1} \left(\frac{\eps}{\log n}\right)^{2^{1/\delta+4}} \textrm{.}
\end{equation}
	
\noindent
By assumption, $\eps \geq \cutoff$ $\geq n^{-\delta/\exp(1/\delta)}$.
For sufficiently small $\delta$, $\delta/\exp(1/\delta) < 2\delta/2^{1/\delta+4}$.
Thus, $\eps \geq (\log n) n^{-2\delta/(2^{1/\delta+4})}$.
Plugging into the RHS of the previous equation, $||\trunvec{s}{R_{1/\delta+2}}{1/\delta+3}\|^2_2 \geq 1/n^{1+2\delta} = 1/n^{\delta(1/\delta+2)}$.
This implies that $v \in S_{1/\delta+2}$, a contradiction.
%     The last inequality holds for all $\eps > \eps_{\sf CUTOFF}$ and follows from the fact that $\delta$ is a very small constant and hence $\frac{2^{1/\delta + 1} + 2}{2\delta} \leq \frac{\log n}{\log \log n}$ for sufficiently large $n$ implying that $(\log n)^{2^{1/\delta + 1} + 2} \leq n^{2\delta}$.  
\end{proof}

\subsection{The correlation lemma}

The following lemma is an important tool in our analysis. Here is an intuitive explanation.
Fix some $s \in S_i$. By \Prop{trun}, the probability $\trun{s}{v}{R_i}{i+1}$  is the correlation
between the vectors $\trunvec{s}{R_i}{i}$ and $\trunvec{v}{R_i}{i}$. If many of these probabilities
are large, then there are many $v$ such that $\trunvec{v}{R_i}{i}$ is correlated with $\trunvec{s}{R_i}{i}$.
We then expect many of these vectors are correlated among themselves.

\begin{definition} \label{def:dist} For $s \in R_i$, the distribution $\cD_{s, i}$
has support $R_i$, and the probability of $u \in R_i$ is $\trunp{s}{v}{R_i}{i+1} = \trun{s}{v}{R_i}{i+1}/\|\trunvec{s}{R_i}{i+1}\|_1$.
\end{definition}

\begin{lemma} \label{lem:corr} Fix arbitrary $s \in R_i$.
$$ \EX_{u_1, u_2 \sim \cD_{s,i}} [\trunvec{u_1}{R_i}{i} \cdot \trunvec{u_2}{R_i}{i}] \geq \frac{1}{\|\trunvec{s}{R_i}{i+1}\|^2_1} \cdot\frac{\|\trunvec{s}{R_i}{i+1}\|^4_2}{\|\trunvec{s}{R_i}{i}\|^2_2} $$
\end{lemma} 

\begin{proof} 
\begin{align}
    & \EX_{u_1, u_2 \sim \cD_{s,i}} [\trunvec{u_1}{R_i}{i} \cdot \trunvec{u_2}{R_i}{i}] \\
     =  &\sum_{u_1, u_2 \in R_i} \|\trunvec{s}{R_i}{i+1}\|^{-2}_1 \trun{s}{u_1}{R_i}{i+1} \trun{s}{u_2}{R_i}{i+1}
    \trunvec{u_1}{R_i}{i} \cdot \trunvec{u_2}{R_i}{i} \\
    = &\|\trunvec{s}{R_i}{i+1}\|^{-2}_1  \sum_{u_1, u_2 \in R_i} (\trunvec{s}{R_i}{i} \cdot \trunvec{u_1}{R_i}{i})
    (\trunvec{s}{R_i}{i} \cdot \trunvec{u_2}{R_i}{i}) (\trunvec{u_1}{R_i}{i} \cdot \trunvec{u_2}{R_i}{i}) \ \ \ \ \textrm{(\Prop{trun})} \\
	=  &\|\trunvec{s}{R_i}{i+1}\|^{-2}_1  \sum_{u_1, u_2 \in R_i} (\trunvec{s}{R_i}{i} \cdot \trunvec{u_1}{R_i}{i})
    (\trunvec{s}{R_i}{i} \cdot \trunvec{u_2}{R_i}{i}) \sum_{w \in R_i} \trun{u_1}{w}{R_i}{i} \trun{u_2}{w}{R_i}{i})  \\
    =  & \|\trunvec{s}{R_i}{i+1}\|^{-2}_1  \sum_{w \in R_i} \sum_{u_1, u_2 \in R_i} [(\trunvec{s}{R_i}{i} \cdot \trunvec{u_1}{R_i}{i}) \trun{u_1}{w}{R_i}{i}][(\trunvec{s}{R_i}{i} \cdot \trunvec{u_2}{R_i}{i}) \trun{u_2}{w}{R_i}{i}] \\
    =  &\|\trunvec{s}{R_i}{i+1}\|^{-2}_1 \sum_{w \in R_i} \left[\sum_{u \in R_i} (\trunvec{s}{R_i}{i} \cdot \trunvec{u}{R_i}{i}) \trun{u}{w}{R_i}{i}\right]^2 \label{eq:exp}
\end{align}

We now write out $\|\trunvec{s}{R_i}{i+1}\|^2_2$ $= \sum_{u \in R_i} \trun{s}{u}{R_i}{i+1}^2$
$= \sum_{u \in R_i} (\trunvec{s}{R_i}{i}\cdot \trunvec{u}{R_i}{i})^2$, by \Prop{trun}.
We expand further below. The only inequality is Cauchy-Schwartz.
\begin{eqnarray}
    \|\trunvec{s}{R_i}{i+1}\|^2_2 & = & \sum_{u \in R_i} (\trunvec{s}{R_i}{i}\cdot \trunvec{u}{R_i}{i}) \sum_{w \in R_i} \trun{s}{w}{R_i}{i} \trun{u}{w}{R_i}{i} \\
    & = & \sum_{w \in R_i} \trun{s}{w}{R_i}{i} \Big[\sum_{u \in R_i} (\trunvec{s}{R_i}{i}\cdot \trunvec{u}{R_i}{i})\trun{u}{w}{R_i}{i} \Big] \\
    & \leq & \sqrt{\sum_{w \in R_i} \trun{s}{w}{R_i}{i}^2} \sqrt{\sum_{w \in R_i} \Big[\sum_{u \in R_i} (\trunvec{s}{R_i}{i}\cdot \trunvec{u}{R_i}{i})\trun{u}{w}{R_i}{i} \Big]^2} \\
    & = & \|\trunvec{s}{R_i}{i}\|_2 \|\trunvec{s}{R_i}{i+1}\|_1 \sqrt{\EX_{u_1, u_2 \sim \cD_{s_i}} [\trunvec{u_1}{R_i}{i} \cdot \trunvec{u_2}{R_i}{i}]} \ \ \ \ \textrm{(by \Eqn{exp})}
\end{eqnarray}
We rearrange and take squares to complete the proof.
\end{proof}

We can apply previous norm bounds to get an explicit lower bound. To see the significance of the following lemma,
note that by \Clm{strata} and Cauchy-Schwartz, $\forall u_1, u_2 \in R_i$, $\trunvec{u_1}{R_i}{i} \cdot \trunvec{u_2}{R_i}{i} \leq 1/n^{\delta(i-1)}$
(fairly close to the lower bound below).

\begin{lemma} \label{lem:corr-bound} Fix arbitrary $s \in S_i$.
$$ \EX_{u_1, u_2 \sim \cD_{s,i}} [\trunvec{u_1}{R_i}{i} \cdot \trunvec{u_2}{R_i}{i}] \geq 1/n^{\delta(i+1)} $$
\end{lemma}

\begin{proof} By \Lem{corr}, the LHS is at least $\frac{1}{\|\trunvec{s}{R_i}{i+1}\|^2_1} \cdot\frac{\|\trunvec{s}{R_i}{i+1}\|^4_2}{\|\trunvec{s}{R_i}{i}\|^2_2}$.
Note that $\|\trunvec{s}{R_i}{i+1}\|_1 \leq 1$. By \Def{strata}, $\|\trunvec{s}{R_i}{i+1}\|^2_2 \geq 1/n^{\delta i}$.
Since $s \in S_i \subseteq R_i$, by \Clm{strata}, $\|\trunvec{s}{R_i}{i}\|^2_2 \leq 1/n^{\delta(i-1)}$.
\end{proof}

\section{Analysis of \findclique} \label{sec:path}

This is the central theorem of our analysis. It shows that the \findclique$(s)$ procedure
discovers a $K_{r^2, r^2}$ minor with non-trivial probability when $s$ is in a sufficiently high stratum.

\begin{theorem}\label{thm:findclique_succeeds_probability}
Suppose $s \in S_i$, for  $\threshold \leq i \leq 1/\delta + 3$.
The probability that the paths discovered in \findclique$(s)$ contain a $K_{r^2, r^2}$ minor
is at least $n^{-4\delta r^4}$.
% The probability that \findclique algorithm on input a graph $G = (V,E)$ with $\Delta(G) = d$ which is $\eps$-far from being $H$-minor free succesfully finds a $K_{r^2 \times r^2}$ minor in $G$ (and therefore an $H$-minor) is at least .
% provided the initial vertex $s$ given to \findclique belongs to a strata with index $i$ where.
\end{theorem}

\Thm{findclique_succeeds_probability} is proved in \Sec{findclique:succeeds:probability}. Towards
the proof, we will need multiple tools. In \Sec{findpath}, we perform a standard calculation to bound
the success probability of \findpath. In \Sec{findclique}, we use this bound to show
that the sets $A$ and $B$ sampled by \findclique{} are successfully connected by paths
as discovered by \findpath. In \Sec{minor}, we argue that the intersections of these paths
is ``well-behaved" enough to induce a $K_{r^2,r^2}$ minor.

We note that the $\sqrt{n}$ in the final running time comes from the calls to \findpath{} in \findclique.

\subsection{The procedure \findpath} \label{sec:findpath}
For convenience, we reproduce the procedure \findpath. 
It is a relatively straightforward application of a birthday paradox argument for bidirectional path finding.

\medskip 
\noindent
\fbox{
\begin{minipage}{0.9\textwidth}
{\tt \findpath$(u,v,k,i)$}

\smallskip
\begin{compactenum}
    \item Perform $k$ random walks of length $2^i\ell$ from $u$ and $v$.
    \item If a walk from $u$ and $v$ terminate at the same vertex, return these paths.
\end{compactenum}
\end{minipage}}

\begin{lemma} \label{lem:path} Let $c$ be a sufficiently large constant.
Consider $u, v \in R_i$. Suppose there exist $\alpha \leq \beta$ such that
$\max(\|\trunvec{u}{R_i}{i}\|^2_2, \|\trunvec{v}{R_i}{i}\|^2_2) \leq 1/n^\alpha$ and $\trunvec{u}{R_i}{i} \cdot \trunvec{v}{R_i}{i} \geq 1/2n^\beta$.
Then, with $k \geq c  n^{\beta/2 + 4(\beta-\alpha)}$, \findpath$(u,v,k,i)$ returns an $R_i$-returning path 
of length $2^{i+1}\ell$ with probability $\geq 2/3$.
\end{lemma}

\begin{proof} First, define $W = \{w | \trun{u}{w}{R_i}{i}/\trun{v}{w}{R_i}{i} \in [1/(8n^{\beta-\alpha}), 8n^{\beta-\alpha}]\}$.
\begin{align*}
	&\sum_{w \notin W} \trun{u}{w}{R_i}{i} \trun{v}{w}{R_i}{i} \leq (8n^{\beta-\alpha})^{-1} \sum_{w \notin W} \max(\trun{u}{w}{R_i}{i}, \trun{v}{w}{R_i}{i})^2 \\
	\leq &(8n^{\beta-\alpha})^{-1} (\|\trunvec{u}{R_i}{i}\|^2_2 + \|\trunvec{v}{R_i}{i}\|^2_2) \leq 1/4n^\beta
\end{align*} 
Therefore, $\sum_{w \in W} \trun{u}{w}{R_i}{i} \trun{v}{w}{R_i}{i} \geq 1/2n^{\beta}$. 

For $a,b \leq k$, let $X_{a,b}$ be the indicator for the following event: the $a$th $2^i\ell$-length random walk from $u$ is an $R_i$-returning
walk that ends at some $w \in W$, and the $b$th random walk from $v$ is also $R_i$-returning, ending at the same $w$.
Let $X = \sum_{a,b \leq k} X_{a,b}$.
Observe that the probability that \findpath$(u,v,k,i)$ returns a path is at least 
$\Pr[X > 0]$.

We can bound $\EX[\sum_{a,b\leq k} X_{a,b}] = k^2 \sum_{w \in W} \trun{u}{w}{R_i}{i} \trun{v}{w}{R_i}{i} \geq k^2/4n^{\beta} \geq (c^2/4)n^{4(\beta-\alpha)}$.
Let us now bound the variance. First, let us expand out the expected square.
\begin{equation}
\EX[(\sum_{a,b} X_{a,b})^2]
 =  \sum_{a,b} \EX[X^2_{a,b}] + 2 \sum_{a \neq a', b} \EX[X_{a,b}X_{a',b}] + 2 \sum_{a, b \neq b'} \EX[X_{a,b}X_{a',b}]
+ 2 \sum_{a \neq a', b \neq b'} \EX[X_{a,b}X_{a',b'}] \label{eq:var}
\end{equation}
Observe that $X^2_{a,b} = X_{a,b}$. Furthermore, for $a \neq a', b \neq b'$, by independence of the walks,
$\EX[X_{a,b}X_{a',b'}] = \EX[X_{a,b}] \EX[X_{a',b'}]$.  (This term will cancel out in the variance.)
By symmetry, $\sum_{a \neq a', b} \EX[X_{a,b}X_{a',b}] \leq k^3 \EX[X_{1,1}X_{2,1}]$ (and analogously
for the third term in \Eqn{var}).
Plugging these in and expanding out the $\EX[X]^2$,
$$ \var[X] \leq \EX[X] + 2k^3 \EX[X_{1,1}X_{2,1}] + 2k^3 \EX[X_{1,1}X_{1,2}]$$
Note that $X_{1,1}X_{2,1} = 1$ when the first and second walks from $u$ end
at the same vertex where the first walk from $v$ ends. Thus, $\EX[X_{1,1}X_{2,1}] = \sum_{w \in W} \trun{u}{w}{R_i}{i}^2 \trun{v}{w}{R_i}{i}$.
Since $w \in W$, $\trun{v}{w}{R_i}{i} \leq 8n^{\beta-\alpha} \trun{u}{w}{R_i}{i}$.
Plugging this bound in,
\begin{align}
    2k^3 \EX[X_{1,1}X_{2,1}] \leq 16k^3n^{\beta-\alpha} \sum_{w \in W} \trun{u}{w}{R_i}{i}^3 & \leq  16k^3 n^{\beta-\alpha} [\sum_{w \in W} \trun{u}{w}{R_i}{i}^2]^{3/2} \\
     & =  16n^{\beta-\alpha}[k^2 \sum_{w \in W} \trun{u}{w}{R_i}{i}^2]^{3/2} \\
     & \leq  64n^{2(\beta-\alpha)} [k^2 \sum_{w \in W} \trun{u}{w}{R_i}{i} \trun{v}{w}{R_i}{i}]^{3/2} \\
     & \leq (\EX[X]^{1/2}/(c/128))(\EX[X]^{3/2}) = \EX[X]^2/(c/128)
\end{align}
(For the last line, we use the bound $\EX[X] \geq (c^2/4) n^{4(\beta-\alpha)}$.
We get an identical bound for $2k^3 \EX[X_{1,1}X_{1,2}]$. Putting it all together, we can prove that $\var[X] \leq 4\EX[X]^2/c'$.
An application of Chebyshev proves that $\Pr[X > 0] > 2/3$.
\end{proof}

\subsection{The procedure \findclique} \label{sec:findclique}
For convenience, we reproduce \findclique. 

\noindent
\fbox{
	\begin{minipage}{0.9\textwidth}
		{\tt \findclique$(s)$}
		
		\smallskip
		\begin{compactenum}
			\item For $i = 5r^4, \ldots, 1/\delta + 4$:
			\begin{compactenum}
				\item Perform $2r^2$ independent random walks of length $2^{i+1} \ell$ from $s$. Let the destinations of the first $r^2$
                walks be multiset $A$, and the destinations of the remaining walks be $B$.
				\item For each $a \in A$, $b \in B$:
				\begin{compactenum}
					\item Run \findpath$(a,b,n^{\findpathexp},i)$
				\end{compactenum}
				\item If all calls to \findpath{} return a path, then let the collection of paths be the subgraph $F$.
                Run $\RS(F,H)$. If it returns an $H$-minor, output that and terminate.
			\end{compactenum}
		\end{compactenum}
\end{minipage}}

\begin{lemma} \label{lem:findclique} Suppose $s \in S_i$, for some $i \leq 1/\delta + 4$. 
Condition on the event that $A, B \subseteq R_i$, during the $i$th iteration in \findclique$(s)$.
With probability $(4n^{2\delta})^{-r^4}$,
the calls to \findpath{} output paths from every $a \in A$ to every $b \in B$, where each 
path is an $R_i$-returning walk of length $2^{i+1}\ell$.
\end{lemma}

\begin{proof} The probability that a $2^{i+1}\ell$-length random walk from $s$ ends at $u$
is at least 
$\trun{s}{u}{S_i}{i+1}$ $= \trunp{s}{u}{R_i}{i+1}\|\trunvec{s}{R_i}{i+1}\|_1$.
In the rest of the proof, let $t = |A| = |B| = r^2$ denote the common size of the multisets $A$ and $B$.
For any $a, b \in V$, let $\tau_{a,b}$ be the probability that \findpath$(a,b,n^{\findpathexp},i)$
succeeds in finding an $R_i$-returning walk between $a$ and $b$ (of length $2^{i+1}\ell$). 
The probability of success for \findclique$(s)$ conditioned on $A, B \subseteq R_i$
is at least
\begin{eqnarray}
     \sum_{A \in R^t_i} \sum_{B \in R^t_i} \prod_{a \in A} \trunp{s}{a}{R_i}{i+1} \prod_{b \in B} \trunp{s}{b}{R_i}{i+1} \tau_{a,b} 
     & = & \sum_{A \in R^t_i} \sum_{B \in R^t_i} \prod_{a \in A} \trunp{s}{a}{R_i}{i+1} \Big(\prod_{b \in B} \trunp{s}{b}{R_i}{i+1}\Big)
     \Big( \prod_{b \in B} \tau_{a,b} \Big) \\
     & = & \sum_{B \in R^t_i} \Big(\prod_{b \in B} \trunp{s}{b}{R_i}{i+1}\Big)  \sum_{A \in R^t_i} \prod_{a \in A} \Big[\trunp{s}{a}{R_i}{i+1}
\Big( \prod_{b \in B} \tau_{a,b} \Big) \Big] \nonumber \\
     & = & \sum_{B \in R^t_i} \prod_{b \in B} \trunp{s}{b}{R_i}{i+1} \Big(\sum_{a \in R_i} \trunp{s}{a}{R_i}{i+1} \prod_{b \in B} \tau_{a,b}\Big)^t \textrm{.}
\end{eqnarray}
Observe that $\prod_{b \in B} \trunp{s}{b}{R_i}{i+1}$ is a probability distribution over $B$. By Jensen, we lower bound.
\begin{equation}
\sum_{B \in R^t_i} \Big(\prod_{b \in B} \trunp{s}{b}{R_i}{i+1}\Big) \Big(\sum_{a \in R_i} \trunp{s}{a}{R_i}{i+1} \prod_{b \in B} \tau_{a,b}\Big)^t 
\geq \Big[\sum_{B \in R^t_i} \Big(\prod_{b \in B}  \trunp{s}{b}{R_i}{i+1}\Big) \sum_{a \in R_i} \trunp{s}{a}{R_i}{i+1} \prod_{b \in B} \tau_{a,b}\Big]^t
\end{equation}
We manipulate and expand further.
\begin{eqnarray}
    & & \Big[\sum_{B \in R^t_i} \Big(\prod_{b \in B}  \trunp{s}{b}{R_i}{i+1}\Big) \sum_{a \in R_i} \trunp{s}{a}{R_i}{i+1} \prod_{b \in B} \tau_{a,b}\Big]^t \\
    & = & \Big[\sum_{a \in R_i} \sum_{B \in R^t_i} \trunp{s}{a}{R_i}{i+1} \Big(\prod_{b \in B} \trunp{s}{b}{R_i}{i+1}\Big) \Big(\prod_{b \in B} \tau_{a,b}\Big)\Big]^t \\
    & = & \Big[\sum_{a \in R_i} \trunp{s}{a}{R_i}{i+1} \sum_{B \in R^t_i} \prod_{b \in B} \trunp{s}{b}{R_i}{i+1} \tau_{a,b}\Big]^t \\
    & = & \Big[\sum_{a \in R_i} \trunp{s}{a}{R_i}{i+1} \big(\sum_{b \in R_i} \trunp{s}{b}{R_i}{i+1} \tau_{a,b}\big)^t \Big]^t \\
    & \geq & \Big[ \sum_{a \in R_i} \sum_{b \in R_i} \trunp{s}{a}{R_i}{i+1} \trunp{s}{b}{R_i}{i+1} \tau_{a,b}\Big]^{t^2} \ \ \ \ \textrm{(Jensen)} \\
    & = & \Big[ \EX_{a, b \sim \cD_{s,i}}[\tau_{a,b}] \Big]^{t^2} \textrm{.} \label{eq:clique}
\end{eqnarray}
Towards lower bounding $\tau_{a,b}$, we first lower bound $\trunvec{a}{R_i}{i}\cdot \trunvec{b}{R_i}{i}$.
By \Lem{corr-bound}, $\EX_{a,b}[\trunvec{a}{R_i}{i}\cdot \trunvec{b}{R_i}{i}] \geq 1/n^{\delta(i+1)}$.
Applying Cauchy-Schwartz, $\trunvec{a}{R_i}{i}\cdot \trunvec{b}{R_i}{i} \leq 1/n^{\delta(i-1)}$. 
Let $p$ be the probability (over $a,b$) that $\trunvec{a}{R_i}{i}\cdot \trunvec{b}{R_i}{i} \geq 1/2n^{\delta(i +1)}$.
$$ 1/n^{\delta(i +1)} \leq \EX_{a,b}[\trunvec{a}{R_i}{i}\cdot \trunvec{b}{R_i}{i}] \leq (1-p)/2n^{\delta(i+1)} + p/n^{\delta(i-1)} $$
Thus, $p \geq 1/2n^{2\delta}$.

By \Clm{strata}, for every $a \in R_i$, $\|\trunvec{a}{R_i}{i}\|^2_2 \leq 1/n^{\delta(i-1)}$ (similarly for $b \in R_i$).
Suppose $\trunvec{a}{R_i}{i}\cdot \trunvec{b}{R_i}{i} \geq 1/2n^{\delta(i+1)}$. Let us apply \Lem{path},
with $\alpha = \delta(i-1)$ and $\beta = \delta(i+1)$. The number of paths taken in \findpath{}
(the value $k$) is $n^{\findpathexp}$. Note that $\findpathexp > \delta(i+1)/2 + 8\delta = \beta/2 + 4(\alpha-\beta)$.
By \Lem{path}, in this case, $\tau \geq 1/2$. As argued in the previous paragraph, this will
happen with probability $1/2n^{2\delta}$ (over the choice of $a,b \sim \cD_{s,i}$).
We plug in \Eqn{clique} and deduce that the probability of success is at least $(1/4n^{2\delta})^{r^4}$.
\end{proof}

\subsection{Criteria for \findclique{} to reveal a minor} \label{sec:minor}

Fix $s \in S_i$, as in \Lem{findclique}. This lemma only asserts
that all pairs in $A \times B$ are connected by \findclique{} (with non-trivial probability).
We need to argue that these paths will actually induce a $K_{r^2, r^2}$-minor. 

As in \Lem{findclique}, let us focus on the $i$th iteration within \findclique,
and condition on $A, B \in R_i$. For every $a \in A, b \in B$,
there is a call to \findpath$(a,b,n^{\findpathexp},i)$. Within each such call,
a set of walks is performed from both $a$ and $b$, with the hope of connecting $a$ to $b$. 
We use $a,a'$ (resp. $b,b'$) to refer to elements in $A$ (resp. $B$).

\begin{asparaitem}
    \item Let $\walkset{a}{b}$ be the set of walks from $a$ performed in the call
    to \findpath$(a,b,n^{\findpathexp},i)$ that are $R_i$-returning. We stress that these walks do not necessarily
    end at $b$, and come from a distribution independent of $b$ (but we wish to track the specific call of \findpath{}
    where these walks were performed).
    Note that $\walkset{b}{a}$
    is the set of $R_i$-returning walks from $b$, performed in the same call.
     
    We use $\walkset{a}{}$ to denote the set of all vertices in $\bigcup_{b \in B} \walkset{a}{b}$.
    \item Let $\conpath{a}{b}$ be a single path from $a$ to $b$ discovered by 
    \findpath$(a,b,n^{\findpathexp},i)$, that consists of a walk in $\walkset{a}{b}$
    and a walk $\walkset{b}{a}$ that end at the same vertex.
    If there are many possible such paths, pick the lexicographically least.
\end{asparaitem}

\medskip

Note that any of the paths/sets described above could be empty. We will think of paths
as sequences, rather than sets, since the order in which the path is constructed is relevant.
For any path $P$, we use $P(t)$ to denote the $t$th element in the sequence.
We use $P(\geq t)$ to denote the sequence of elements with index at least $t$.
When we refer to intersections of paths being empty/non-empty, we refer to sets induced
by the corresponding sequence.

For $s \in S_i$, conditioned
on $A, B \subseteq R_i$, \Lem{findclique} gives a lower bound on $\Pr[\bigcap_{a \in A, b \in B} \conpath{a}{b} \neq \emptyset]$.
We will define some \emph{bad} events that interfere with minor structure. 

Recall that $A$ and $B$ are multisets. (It is convenient to think of them as sequences.)
The same vertex may appear multiple times in $A \cup B$,
but we think of each occurrence as a distinct multiset element. Therefore, equality refers to vertex
at the same index in $A$ (or $B$). By definition, elements in $A$ are disjoint from $B$.

\begin{definition} \label{def:bad} The following events are referred to as \emph{bad events} of Type 1, 2, or 3.
We set $\tau = 2^{i-1}\ell$.
\begin{asparaenum}
    \item $\exists a, b, c \in A \cup B$, $c \neq a, b$, such that $\walkset{c}{} \cap \conpath{a}{b} \neq \emptyset$.
    \item $\exists a, b, b'$ (all distinct) such that $\exists W \in \walkset{a}{b}$ where
    $W(\geq \tau) \cap \conpath{a}{b'} \neq \emptyset$. (Or, $\exists a,a' \in A, b \in B$, all distinct,
    such that $\exists W \in \walkset{b}{a}$ where $W(\geq \tau) \cap \conpath{a'}{b} \neq \emptyset$.)
    \item $\exists a, b, W_a \in \walkset{a}{b}, W_b \in \walkset{b}{a}$ such that $W_a, W_b$ end
    at the same vertex and $\exists t_1, t_2$ such that $\min(t_1, t_2) \leq \tau$ and $W_a(t_1) = W_b(t_2)$.
\end{asparaenum}
\end{definition}

For clarity, let us express the above bad events in plain English. Note that $\tau$ is the index
of the midpoint of the walks, so it splits walks into halves.
\begin{asparaenum}
    \item A walk from $c \in A \cup B$ intersects $\conpath{a}{b}$, where $c \neq a,b$.
    \item The second half of a walk in $\walkset{a}{b}$ (which starts from $a$) 
    intersects $\conpath{a}{b'}$ for $b \neq b'$.
    \item A walk in $\walkset{a}{b}$ and a walk in $\walkset{b}{a}$ intersect twice. Note that this
    is a pair of walks, one from $a$ and the other from $b$.
    The first intersection is in the first half of either of the walks. The walks
    also end at the same vertex.
\end{asparaenum}

\begin{claim} \label{clm:minor} If all $\conpath{a}{b}$ sets are non-empty and there
is no bad event, then $\bigcup_{a,b} \conpath{a}{b}$ contains a $K_{r^2, r^2}$-minor.
\end{claim}

\begin{proof} The $\conpath{a}{b}$s
may not form simple paths, and it will be convenient to ``clean them up".
Each $\conpath{a}{b}$ is formed by $W_a \in \walkset{a}{b}$ and $W_b \in \walkset{b}{a}$
that end at the same vertex. Since there is no Type 3 bad event, $W_a(\leq \tau)$ is
disjoint from $W_b$ (and vice versa). Therefore (by removing self-intersections and loops),
we can construct a simple path from $a$ to $b$ with the following (vertex) disjoint contiguous simple paths:
$Q_{a,b} \subseteq W_a(\leq \tau)$, $\widehat{P_{a,b}} \subseteq W_a(\geq \tau) \cup W_b(\geq \tau)$,
and $Q_{b,a} \subseteq W_b(\leq \tau)$.

In each bullet below, we first make a statement about the disjointness of these various sets. 
The proof follows immediately. We consider $a, a' \in A$ and $b, b' \in B$, where
the elements in $A$ (or $B$) might be equal.

\begin{asparaitem}
    \item If $a \neq a'$, $Q_{a,b} \cap Q_{a',b'} = \emptyset$.
    If $b \neq b'$, $Q_{b,a} \cap Q_{b',a'} = \emptyset$.

    Consider the first statement. (Note that we allow $b = b'$.) Observe that $Q_{a,b} \subseteq \walkset{a}{}$
and $Q_{a',b'} \subseteq \conpath{a'}{b'}$. So $\walkset{a}{} \cap \conpath{a'}{b'} \neq \emptyset$, implying
a Type 1 bad event. The second statement has an analogous proof.

    \item $Q_{a,b} \cap Q_{b',a'} = \emptyset$.

If $a=a', b=b'$, then this holds by the argument in the first paragraph (no Type 3 bad events). Suppose $a\neq a'$.
Then (as before), $Q_{a,b} \subseteq \walkset{a}{}$ and $Q_{b',a'} \subseteq \conpath{a'}{b'}$.
Since no Type 1 bad events occur, $\walkset{a}{} \cap \conpath{a'}{b'} = \emptyset$.
The case $b \neq b'$ is analogous.

    \item If $a \neq a'$ or $b \neq b'$, $\widehat{P_{a,b}} \cap {P_{a',b'}} = \emptyset$.

Wlog, assume $a \neq a'$. Note that $\widehat{P_{a,b}} \subseteq W_a(\geq \tau) \cup W_b(\geq \tau)$,
where $W_a \in \walkset{a}{b}$ and $W_b \in \walkset{b}{a}$. If $W_a(\geq \tau) \cap P_{a',b'} \neq \emptyset$,
then $\walkset{a}{} \cap P_{a',b'} \neq \emptyset$ (a Type 1 bad event). Suppose $W_b(\geq \tau) \cap P_{a',b'} \neq \emptyset$.
If $b \neq b'$, this is Type 1 bad event. So suppose $b = b'$, so $W_b(\geq \tau) \cap P_{a',b} \neq \emptyset$.
Since $W_b \in \walkset{b}{a}$ (for $a \neq a'$), this is Type 2 bad event. 
\end{asparaitem}

\medskip

We construct the minor. Let $C(a) = \bigcup_{b \in B} Q_{a,b}$ and $C(b) = \bigcup_{a \in A} Q_{b,a}$.
Each $C(a), C(b)$ forms a connected subgraph. By the disjointness properties of the $Q_{a,b}$ sets, all the $C(a), C(b)$
sets/subgraphs are vertex disjoint. Note that $\widehat{P_{a,b}}$ is disjoint from all other $P_{a',b'}$ paths \emph{and}
all the $C(a), C(b)$ sets. (We construct $P_{a,b}$ to be disjoint from $Q_{a,b}$ and $Q_{b,a}$ in the first paragraph.
Every other $Q_{a',b'}$ is contained in $P_{a',b'}$.) Thus, we have disjoint paths from each $C(a)$ to $C(b)$,
which gives a $K_{r^2,r^2}$-minor.
\end{proof} 

\subsection{The probabilities of bad events} \label{sec:bad}

In this section, we bound the probability of bad events, as detailed in \Def{bad}.
As before, we fix $s \in S_i$ and condition on $A \cup B \subseteq R_i$.

We require some technical definitions of random walk probabilities.

\begin{definition} \label{def:prob} Let $\trw{s}{S}{t}{v}$ be the probability of a walk from $s$
to $v$ of length $t$ being $S$-returning. (We allow $\ell \nmid t$, and require that the walk
encounters $S$ at every $j\ell$th step, for $j \leq \lfloor t/\ell \rfloor$.)

We use $\trwvec{s}{S}{t}$ to denote the vector of these probabilities. More generally,
given any distribution vector $\bx$ on $V$, $\trwvec{\boldsymbol{x}}{S}{t}$
denotes the vector of $S$-returning walk probabilities at time $t$.
\end{definition}

We stress that this is not a conditional probability. Note that if $t = 2^i\ell$,
then $\trwvec{s}{S}{t} = \trunvec{s}{S}{i}$. We show some simple propositions
on these vectors. Let $\mathbb{I}_S$ denote the $n \times n$ matrix that 
preserves all coordinates in $S$ and zeroes out other coordinates.

\begin{proposition} \label{prop:prob} The vector $\trwvec{\boldsymbol{x}}{S}{t}$
evolves according to the following recurrence.
Firstly, $\trwvec{\boldsymbol{x}}{S}{0} = \bx$. For $t \geq 1$ such that $\ell \nmid t$,
$\trwvec{\boldsymbol{x}}{S}{t} = M\trwvec{\boldsymbol{x}}{S}{t-1}$.
For $t \geq 1$ such that $\ell \mid t$, $\trwvec{\boldsymbol{x}}{S}{t} = \mathbb{I}_S M\trwvec{\boldsymbol{x}}{S}{t-1}$
\end{proposition}

\begin{proposition} \label{prop:max} For all $\bx$ and all $t\geq 1$, $\|\trwvec{\bx}{S}{t}\|_\infty \leq \|\trwvec{\bx}{S}{t-1}\|_\infty$.
\end{proposition}

\begin{proof} Since $M$ is a symmetric random walk matrix, it computes the ``new" value at a vertex by averaging the values of the neighbors
(and itself). This can never increase the maximum value. Furthermore, $\mathbb{I}_S$ only zeroes out some coordinates.
This proves the proposition.
\end{proof}

In what follows, we fix the walk length to $2^i\ell$. To reduce clutter,
we drop notational dependencies on this length.

\begin{definition} \label{def:walk} The distribution of $2^i\ell$-length walks from $u$ is denoted $\cW_{u}$.
For any walk $W$, $\walkt{u}{t}$ denotes the $t$th vertex of the walk.

The Boolean predicate $\rho(\walk{u})$ is true if $\walk{u}$ is $R_i$-returning.
\end{definition}

Recall that $\cD_{s,i}$ is the distribution with support $R_i$, where the probability of $u \in R_i$
is $\trunp{s}{v}{R_i}{i+1}/\|\trunvec{s}{R_i}{i+1}\|_1$ (\Def{dist}). Conditioned on $a \in R_i$,
this is precisely the distribution that the elements of the sets $A, B$ are drawn from. Refer to \findclique,
where $A \cup B$ are the destinations of $2^{i+1}\ell$ length random walks from $s$.
Since $i$ is fixed, we will simply write this as $\cD_s$.

\begin{claim} \label{clm:bad1} For any $F \subseteq V$:
\begin{asparaenum}
    \item  $$\Pr_{a \sim \cD_s, \walk{a} \sim \cW_a}[\rho(\walk{a}) \wedge \walk{a} \cap F \neq \emptyset] \leq 2^i\ell|F|/(n^{\delta(i-1)}\|\trunvec{s}{R_i}{i+1}\|_1)\textrm{.}$$
    \item For any $a \in R_i$,
$$ \Pr_{\walk{a} \sim \cW_a}[\exists t \geq \tau \ | \ \rho(\walk{a})  \wedge \walkt{a}{t} \in F] \leq 2^i\ell|F|/n^{\delta(i-2)}$$
\end{asparaenum}
\end{claim}

\begin{proof} We prove the first part. Let $\bx$ be the probability vector corresponding to $\cD_s$.
So $\|\bx\|_\infty = \|\trunvec{s}{R_i}{i+1}\|_\infty/\|\trunvec{s}{R_i}{i+1}\|_1$.
By \Prop{max}, $\forall t \geq 1$, $\|\trwvec{\bx}{R_i}{t}\|_\infty \leq \|\bx\|_\infty$.
sing \Clm{max-prob}, this is at most $1/(n^{\delta(i-1)}\|\trunvec{s}{R_i}{i+1}\|_1)$.
We union bound over $F$ and the walk length.
\begin{align*}
   \Pr_{a \sim \cD_s, \walk{a} \sim \cW_a}[\rho(\walk{a}) \wedge \walk{a} \cap F \neq \emptyset] 
   & \leq \sum_{t \leq 2^i\ell} \sum_{v \in F} \Pr_{a \sim \cD_s, \walk{a} \sim \cW_a}[\rho(\walk{a}) \wedge \walkt{a}{t} = v] \\
   & \leq  \sum_{t \leq 2^i\ell} \sum_{v \in F} \|\bx\|_\infty \leq 2^i\ell|F|/(n^{\delta(i-1)}\|\trunvec{s}{R_i}{i+1}\|_1)
\end{align*}

Now for the second part. By the union bound, the probability is bounded
above by 
\begin{equation}
\sum_{t \geq 2^{i-1}\ell} \sum_{u \in F} \Pr_{\walk{a} \sim \cW_a}[\rho(\walk{a}) \wedge \walkt{a}{t} = u]
\leq \sum_{t \geq 2^{i-1}\ell} \sum_{u \in F} \|\trwvec{a}{R_i}{t}\|_\infty \label{eq:walk1}
\end{equation}
By \Prop{max}, the infinity norm is bounded above by $\|\trwvec{a}{R_i}{2^{i-1}\ell}\|_\infty = \|\trunvec{a}{R_i}{i-1}\|_\infty$.
By \Clm{max-prob}, the latter is at most $1/n^{\delta(i-2)}$.
Plugging in \Eqn{walk1}, we get an upper bound of $2^{i-1}\ell|F|/n^{\delta(i-2)}$.
\end{proof}

\begin{claim} \label{clm:bad3} For any $a \in R_i$,
\begin{eqnarray*}
& & \Pr_{b \sim \cD_s, \walk{a} \sim \cW_a, \walk{b} \sim \cW_b}[\rho(\walk{a}) \wedge \rho(\walk{b}) \wedge 
\walkt{a}{2^i\ell} = \walkt{b}{2^i\ell} \wedge
(\exists t_a, t_b, \min(t_a, t_b) \leq \tau, \  
\walkt{a}{t_a} = \walkt{b}{t_b}) ] \\
& \leq & 2^{2i}\ell^2/(n^{\delta(2i - 2)}\|\trunvec{s}{R_i}{i+1}\|_1)
\end{eqnarray*}
\end{claim}

\begin{proof} Let us write out the main event in English. We fix an arbitrary $a$,
and pick $b \sim \cD_s$. We perform $R_i$-returning walks of length $2^i\ell$ from both $a$ and $b$.
We are bounding the probability that the ``initial half" (less than $2^{i-1}\ell$ steps)
of one of the walks intersects with the other, and subsequently, both walks end at the same vertex.

To that end, let us define two vertices $w_1, w_2$. We want to bound the probability
of that both walks first encounter $w_1$, and then  end at $w_2$. It is be very useful
to treat the latter part simply as two walks from $w_1$, where one of them is at least
of length $2^{i-1}\ell$. Note that $w_1$ might not be in $R_i$.

Let $Z_{a,t}$ be the random variable denoting the $t$th vertex of a random 
walk from $a$. Let us also define $R_i$-returning walks with an offset $g$, starting from $w$.
Basically, such a walk starts from $w$ (that may not be in $R_i$) and performs $g$
steps to end up in $R_i$. Subsequently, it behaves as an $R_i$-returning walk.
Observe that the second parts of the walks are $R_i$-returning walks from $w_1$,
with offsets of $\ell - [t_a(\textrm{mod } \ell)]$, $\ell-[t_b(\textrm{mod } \ell)]$.
Let $Y_{w,t}$ be the random variable denoting the $t$th vertex of an $R_i$-returning
walk from $w$, with the offset $\ell-[t(\textrm{mod } \ell)]$. We use primed versions
for independent such variables.

Let us fix values for $t_a, t_b$ such that $\min(t_a, t_b) \leq \tau = 2^{i-1}\ell$.
(We will eventually union bound over all such values.)  
The probability we wish to bound is the following.
We use independence of the walks to split the probabilities. There are four independent walks under consideration:
one from $a$, one from $b$, and two from $w$.
\begin{eqnarray*}
& & \sum_{w_1 \in V} \sum_{w_2 \in V} \Pr_{b \sim \cD_s, \cW_a, \cW_b, \cW_{w_1}}[Z_{a,t_a} = w_1 \wedge Z_{b,t_b} = w_1  \wedge Y_{w_1,2^i\ell-t_a} = w_2
\wedge Y'_{w_1,2^i\ell-t_b} = w_2] \\
& = &\sum_{w_1 \in V} \sum_{w_2 \in V} \Pr_{\cW_a}[Z_{a,t_a} = w_1] \Pr_{b\sim \cD_s, \cW_b}[Z_{b,t_b} = w_1]
\Pr_{\cW_{w_1}}[Y_{w_1, 2^i\ell-t_a} = w_2] \Pr_{\cW_{w_1}}[Y_{w_1, 2^i\ell-t_b} = w_2] \label{eq:inter}
\end{eqnarray*}
Consider $\Pr_{b\sim \cD_s, \cW_b}[Z_{b,t_b} = w_1]$. This is exactly the $w_1$th
entry in $\trwvec{\bx}{\R_i}{t_b}$ where $\bx$ is the distribution given by $\cD_s$.
By \Prop{max}, this is at most $\|\bx\|_\infty$, which is at most $1/(n^{\delta(i-1)}\|\trunvec{s}{R_i}{i+1}\|_1)$
(as argued in the second pat of the proof of \Clm{bad1}).

Since $\min(t_a,t_b) \leq \tau$, one of $2^i\ell-t_a, 2^i\ell-t_b$ is at least $2^{i-1}\ell$. Thus, one
of $\Pr_{\cW_{w_1}}[Y_{w_1, 2^i\ell-t_a} = w_2]$ or $\Pr_{\cW_{w_1}}[Y_{w_1, 2^i\ell-t_b} = w_2]$
refers to a walk of length at least $2^{i-1}\ell$. Let us bound $\Pr_{\cW_{w_1}}[Y_{w_1, t} = w_2]$
for $t \geq 2^i\ell$. We can break such a walk into two parts: the first $\ell - [t(\textrm{mod } \ell)]$
steps lead to some $v \in R_i$, and the second part is an $R_i$-returning walk of length at least $2^i\ell$
from $v$ to $w$. Recall that $\prw{x}{y}{d}$ is the standard random walk probability
of starting from $x$ and ending at $y$ after $d$ steps. For some $t' \geq 2^i\ell$,
\begin{align*}
\Pr_{\cW_{w_1}}[Y_{w_1, t} = w_2]  = & \sum_{v \in R_i} \prw{w_1}{v}{\ell - [t(\textrm{mod } \ell)]}
\trw{v}{R_i}{t'}{w_2} 
\leq \sum_{v \in R_i} \prw{w_1}{v}{\ell - [t(\textrm{mod } \ell)]} \|\trwvec{v}{R_i}{t'}\|_\infty \\
 \leq & \sum_{v \in R_i} \prw{w_1}{v}{\ell - [t(\textrm{mod } \ell)]} \|\trunvec{v}{R_i}{i}\|_\infty
\leq \sum_{v \in R_i} \prw{w_1}{v}{\ell - [t(\textrm{mod } \ell)]} n^{-\delta(i-1)} = n^{-\delta(i-1)}
\end{align*}

Plugging these bounds in \Eqn{inter}, for fixed $t_a, t_b$, there exists $t \in \{2^i\ell-t_a, 2^i\ell-t_b\}$
such that the probability of the main event is at most
\begin{eqnarray*} & & (1/n^{\delta(i-1)}\|\trunvec{s}{R_i}{i+1}\|_1) \cdot (1/n^{\delta(i-1)})
\sum_{w_1 \in V} \sum_{w_2 \in V} \Pr_{\cW_a}[Z_{a,t_a} = w_1] \Pr_{\cW_{w_1}}[Y_{w_1, t} = w_2] \\
& \leq & 1/(n^{\delta(2i - 2)}\|\trunvec{s}{R_i}{i+1}\|_1) \sum_{w_1 \in V}\Pr_{\cW_a}[Z_{a,t_a} = w_1]
\sum_{w_2 \in V}\Pr_{\cW_{w_1}}[Y_{w_1, t} = w_2] = 1/(n^{\delta(2i - 2)}\|\trunvec{s}{R_i}{i+1}\|_1)
\end{eqnarray*}
A union bound over all pairs of $t_a, t_b$ completes the proof.
\end{proof}

We now bound the total probability of bad events. Most of the technical work is already done
in the previous lemmas; we only need to perform some union bounds.

\begin{lemma} \label{lem:bad-bound} Conditioned on $A \cup B \subseteq R_i$, the total probability of bad events is at most
\begin{equation}
\frac{2^{2i+4} r^8 n^{30\delta}}{n^{\delta i/2} \|\trunvec{s}{R_i}{i+1}\|_1}
\end{equation}
\end{lemma}

\begin{proof} We bound the bad events by type. Recall that $\ell = n^{5\delta}$.

\textbf{Type 1:} $\exists a, b, c \in A \cup B$, $c \neq a, b$, such that $\walkset{c}{} \cap \conpath{a}{b} \neq \emptyset$.

Fix a choice of $a \in A, b \in B$. 
Conditioned in $A \cup B \subseteq R_i$, any $c \neq a, b$ is drawn from $\cD_{s}$. 
In \Clm{bad1}, set $F = \conpath{a}{b}$.
By the first part of \Clm{bad1}, the probability that a single walk drawn from $\cW_c$ is $R_i$-returning and intersects $\conpath{a}{b}$
is at most $2^i\ell(2^{i+1}\ell)/n^{\delta(i-1)}\|\trunvec{s}{R_i}{i+1}\|_1$.
The set $\walkset{c}{}$ consists of at most $r^2 n^{\findpathexp}$ such walks. We union bound over all these walks,
and all $r^4$ choices of $a,b$, and plug in $\ell = n^{5\delta}$ to get an upper bound of
$$\frac{2^{2i+1}\ell^2 r^6 n^{\findpathexp}}{n^{\delta(i-1)} \|\trunvec{s}{R_i}{i+1}\|_1}
= \frac{2^{2i+1}r^6 n^{20\delta}}{n^{\delta i/2}\|\trunvec{s}{R_i}{i+1}\|_1} $$

\textbf{Type 2:} $\exists a, b, b'$ (all distinct) such that $\exists W \in \walkset{a}{b}$ where
$W(\geq \tau) \cap \conpath{a}{b'} \neq \emptyset$. (Or, $\exists a,a' \in A, b \in B$ with analogous conditions.)

Fix $a,b,b'$. Set $F = \conpath{a}{b'}$ in \Clm{bad1}. By the second part of \Clm{bad1},
the probability that a single walk from $\cW_a$ is $R_i$-returning and
intersects $F$ at step $\geq \tau$ is at most $2^i\ell(2^{i+1}\ell)/n^{\delta(i-2)}$.
We union bound over all the $r^2 n^{\findpathexp}$ walks in $\walkset{a}{}$ and all $r^6$ choices
of $a,b,b'$. (We also union bound over choosing $b, b'$ or $a,a'$.) 
The upper bound is $2^{2i+1} r^6 n^{21\delta}/ n^{\delta i/2}$.

\textbf{Type 3:} $\exists a, b, W_a \in \walkset{a}{b}, W_b \in \walkset{b}{a}$ such that $W_a, W_b$ end
at the same vertex and $\exists t_1, t_2$ such that $\min(t_1, t_2) \leq \tau$ and $W_a(t_1) = W_b(t_2)$.

This case is qualitatively different. We will take a union bound over \emph{pairs} of walks, and
require the stronger bound of \Clm{bad3}.

Fix $a \in A$. Observe that $b \sim \cD_s$. For a single walk $W_a \sim \cW_a$ and a single walk $W_b \sim \cW_b$,
the probability of a Type 3 bad event is bounded by \Clm{bad3}. The upper bound 
is $2^{2i}\ell^2/(n^{\delta(2i - 2)}\|\trunvec{s}{R_i}{i+1}\|_1)$. We union bound over the $r^4 n^{\findpathexptwo}$
\emph{pairs} of walks from $a$ and $b$, and then over the $r^4$ choices of $a,b$.
The final bound is:
$$ \frac{2^{2i} r^4 \ell^2 n^{\findpathexptwo}}{n^{\delta(2i - 2)}\|\trunvec{s}{R_i}{i+1}\|_1}
= \frac{2^{2i} r^4 n^{30\delta}}{n^{\delta i} \|\trunvec{s}{R_i}{i+1}\|_1} $$

We complete the proof by taking a union bound over the three types. Note that $\|\trunvec{s}{R_i}{i+1}\|_1 \leq 1$,
so we can upper bound the probability of each type of bad event 
by $\frac{2^{2i+1} r^8 n^{30\delta}}{n^{\delta i/2} \|\trunvec{s}{R_i}{i+1}\|_1}$.
\end{proof}

\subsection{Proof of \Thm{findclique_succeeds_probability}} \label{sec:findclique:succeeds:probability}

\begin{proof} Fix $s \in S_i$. Let $\cC$ be the event that $A \cup B \subseteq R_i$,
let $\cE$ be the event $\bigcap_{a \in A, b \in B} \conpath{a}{b} \neq \emptyset$,
and let $\cF$ be the union of bad events. By \Clm{minor}, the probability that \findclique$(s)$
find a minor is at least $\Pr[\cE \cap \overline{\cF}]$. 
We lower bound as follows: $\Pr[\cE \cap \overline{\cF}]
\geq \Pr[\cC] \Pr[\cE \cap \overline{\cF} | \cC] \geq \Pr[\cC] (\Pr[\cE | \cC] - \Pr[\cF | \cC])$.

Note that $\Pr[\cC] = \|\trunvec{s}{R_i}{i+1}\|_1^{2r^2}$. By \Clm{1-norm}, $\|\trunvec{s}{R_i}{i+1}\|_1 \geq n^{-\delta}$,
so $\Pr[\cC] \geq n^{-2\delta r^2}$.

\Lem{findclique} provides a lower bound for $\Pr[\cE | \cC]$, and \Lem{bad-bound}
provides an upper bound for $\Pr[\cF | \cC]$. We plug these bounds in below.

\begin{eqnarray}
    \Pr[\cE | \cC] - \Pr[\cF | \cC] & \geq & \frac{1}{(4n^{2\delta})^{r^4}} - \frac{2^{2i+4} r^8 n^{30\delta}}{n^{\delta i/2} \|\trunvec{s}{R_i}{i+1}\|_1}
\end{eqnarray}

Observe how the positive term is independent of $i$, while the negative term decays exponentially in $i$.
This is crucial to argue that for a sufficiently large (constant) $i$, the lower bound is non-trivial.

When $i \geq 5r^4$, $n^{i \delta/2} \geq n^{2\delta r^4 + \delta r^4/2} \geq n^{2\delta r^4 + 40\delta}$ (note that,
$r$, the number of vertices in $H$, is at least $3$). By \Clm{1-norm}, $\|\trunvec{s}{R_i}{i+1}\|_1 \geq n^{-\delta}$.
Thus, for sufficiently large $n$, $\Pr[\cF | \cC] \leq 1/(2(4n^{2\delta})^{r^4})$.
Putting it all together, the probability of finding a $K_{r^2,r_2}$-minor is at least $n^{-4\delta r^4}$.
\end{proof}

\section{Local partitioning in the trapped case} \label{sec:partition}

\Thm{findclique_succeeds_probability} tells us that if there are $\Omega(n^{1-\delta})$ vertices
in strata numbered $5r^4$ and above, then \isminorfree{} finds a biclique minor with high probability.
We deal with the case when most vertices lie in low strata, i.e, random walks from most vertices are trapped in a very small subset.

We will argue that (almost) all vertices in low strata can be partitioned
into ``pieces", such that each piece is a low conductance cut, and (a superset of) each piece
can be found by performing random walks in $G$. If \isminorfree{}
fails to find a minor, this lemma can be iteratively applied to make $G$ $H$-minor free
by removing few edges  (this argument is given in \Sec{final}).

We use $\prw{s}{v}{t}$ to denote the probability that at $t$ length random walk
from $s$ ends at $v$.

% 
% The main partitioning lemma follows. We assume the following bounds of the parameters.
% \begin{asparaitem}
%     \item $\alpha = \Omega(\eps)$.
%     \item $\ell = n^{5\delta}$.
% \end{asparaitem}
 
\begin{restatable}{lemma}{partitionlemma} \label{lem:partition} Let $\alpha \geq n^{-\delta/2}$.
Consider some subset $S \subseteq V$ and $i \in \NN$ such
	that $\forall s \in S, \|\trunvec{s}{S}{i}\|^2_2 \leq 1/n^{\delta(i-1)}$. Define $S' \subseteq S$
	to be $\{s | s\in S  \ \textrm{and} \ \|\trunvec{s}{S}{i+1}\|^2_2 \geq 1/n^{\delta i}\}$.
	
	Suppose $|S'| \geq \alpha n$. Then, there is a subset $\widetilde{S} \subseteq S'$, $|\widetilde{S}| \geq \alpha n/8$
	such that for $\forall s \in \widetilde{S}$: there exists a subset $P_s \subseteq S$ where
	\begin{asparaitem}
		\item $E(P_s, S\setminus P_s) \leq 2n^{-\delta/4} d|P_s|$
		\item $\forall v \in P_s$, $\exists t \leq 160n^{\delta(i+7)}/\alpha$ such that $\prw{s}{v}{t} \geq \alpha/n^{\delta(2i+14)}$.
	\end{asparaitem}
\end{restatable}

The aim of this section is to prove this lemma. Henceforth, we will assume that $S, S'$
are as defined in the lemma.

Using the norm bounds, we show that for every vertex $s \in S'$, there is a large set
of destination vertices that are all reached with high probability through
random walks of length $2^{i+1}\ell$.

\begin{claim} \label{clm:set} For every $s \in S'$, there exists a set $U_s \subseteq S$,
$|U_s| \geq n^{\delta(i-2)}$, such that $\forall u \in U_s$, $\prw{s}{u}{2^{i+1}\ell} \geq 1/2n^{\delta i}$.
\end{claim}

\begin{proof} By \Prop{trun}, for any $u \in U$,  $\trun{s}{u}{S}{i+1} = \trunvec{s}{S}{i} \cdot \trunvec{u}{S}{i}$.
By the property of $S$ and Cauchy-Schwartz, $\trun{s}{u}{S}{i+1} \leq 1/n^{\delta(i-1)}$. 

Since $s \in S'$, $\sum_{u \in S} \trun{s}{u}{S}{i+1}^2 \geq 1/n^{\delta i}$.
Let us simply define $U_s$ to be $\{u | u \in S, \trun{s}{u}{S}{i+1} \geq 1/2n^{\delta i}\}$.
Note that $\prw{s}{u}{2^{i+1}\ell} \geq \trun{s}{u}{S}{i+1}$.
\begin{eqnarray}
 1/n^{\delta i} \leq \sum_{u \in S} \trun{s}{u}{S}{i+1}^2 
& = & \sum_{u \in U_s} \trun{s}{u}{S}{i+1}^2 + \sum_{u \notin U_s}\trun{s}{u}{S}{i+1}^2 \\
&\leq & |U_s|/n^{2\delta(i-1)} + (1/2n^{\delta i})\sum_{u \notin U_s} \trun{s}{u}{S}{i+1}
\leq |U_s|/n^{2\delta(i-1)} + 1/2n^{\delta i}
\end{eqnarray}
We rearrange to bound the size of $U_s$.
\end{proof}

%Intuitively, \Clm{set} sets up the stage for establishing \Lem{partition}. Indeed, if there is a large set $U_s$ of ``favorite vertices'' anchored to every vertex $s \in S'$ and $S'$ is indeed large as stated in the premise, then perhaps $S''$ can be greedily chosen from $S'$ to include those vertices $s$ the favorite set anchored to which is reached with large probability. This is made formal below in \Sec{projected}.

\subsection{Local partitioning on the projected Markov chain} \label{sec:projected}

We define the ``projection" of the random walk onto the set $S$.
This uses a construction of ~\cite{KalePS:13}.
We define a Markov chain $M_S$ over the set $S$. We retain all transitions
from the original random walk on $G$ that are within $S$, and we denote these by $e^{(1)}_{u, v}$ for every $u$ to $v$ transition in the random walk on $G$.
Additionally, for every $u,v \in S$ and $t \geq 2$,
we add a transition $e^{(t)}_{u,v}$. The probability of this transition
is equal to the total probability of $t$-length walks in $G$ from $u$
to $v$, where all internal vertices in the walk lie outside $S$.

Since $G$ is irreducible and the stationary mass on $S$ is non-zero, all walks eventually reach $S$. Thus the outgoing transition probabilities from each $v$ in $M_S$ sum to 1, and hence $M_S$ is a valid Markov chain. 
Furthermore, by the symmetry of the original random walk, $e^{(t)}_{u,v} = e^{(t)}_{v,u}$.
% Since for any $t$ length walk from $u$ to $v$, all vertices on which fall outside $S$, we added
% a transition $e^{(t)}_{u,v}$, it can be seen that the transition probabilities for all edges which end at $v$ in $M_S$ sum to 1. 
Therefore the transition matrix of $M_S$ remains symmetric, and the stationary distribution is uniform on $S$. 

For a transition $e^{(t)}_{u, v}$ in $M_S$, we define the length of this transition to be $t$. 
For clarity, we use ``hops" to denote the length of a walk in $M_S$,
and retain ``length" for walks in $G$. The length of an $h$ hop random walk in $M_S$ is defined to be the sum of the lengths of the transitions it takes.
We note that these ideas come from the work of Kale-Peres-Seshadhri to analyze random
walks in noisy expanders~\cite{KalePS:13}.

We use $\projw{s}{h}$ to denote the distribution of the $h$-hop walk from $s$,
and $\projwp{s}{v}{h}$ to denote the corresponding probability of reaching $v$.
We use $\walkdist{h}$ to denote the distribution of $h$-hop walks starting from the uniform distribution in $S$.

We state Kac's formula (Corollary 24 in Chapter 2 of~\cite{AF-book}, restated).

\begin{lemma} \label{lem:kac-org} (Kac's formula) The expected return time (in $G$) to $S$ of a random walk
starting from $S$ is reciprocal of the fractional stationary mass of $S$, ie $n/|S|$.
\end{lemma}

The following is a direct corollary.

\begin{lemma} \label{lem:kac} $\EX_{W\sim \walkdist{h}}[\textrm{length of $W$}] = hn/|S|$
\end{lemma}

\begin{proof} Since the walk starts at the stationary distribution, it remains in this distribution
at all hops. By linearity of expectation, it suffices to get the expected length for the first hop (and multiply
with $h$). This is precisely expected return time to $S$, if we performed
random walks in $G$. By Kac's formula above, the expected return time to $S$ equals the 
reciprocal of the stationary mass of $S$, which is just $n/|S|$.
\end{proof}

The next lemma is an analogue of \Clm{set} for $M_S$. Recall that $\ell = n^{5\delta}$.
 
% 
% The next lemma is an important stepping stone towards showing \Lem{partition}. We know
% that most of the vertices in the original graph $G$ belong to some lower strata. \Lem{good} 
% confirms that many vertices in $M_S$ are \emph{good}, in the sense that walks with few 
% $(\leq n^{\delta})$ hops from a good vertex \emph{do not mix} over $M_S$ and will have large 
% probability of occupying some vertex. \Sec{proof:of:pl} combines this lemma with some local 
% partitioning tools to show that every good vertex in $S$ one could attach a low conductance 
% piece (in $M_S$) every vertex in which can be reached with probability at least 
% $\alpha/n^{O(\delta r^4)}$ as desired in the statement of \Lem{partition}. Contrast this 
% with \Clm{set} where we showed a statement similar in spirit for random walks on the original 
% graph.

\begin{lemma} \label{lem:good} There exists a subset $S'' \subseteq S'$, $|S''| \geq |S'|/2$,
such that $\forall s \in S''$, $\|\projw{s}{n^{\delta}}\|_\infty \geq 1/n^{\delta(i+6)}$.
\end{lemma}

\begin{proof} Define event $\cE_{s,v,h}$ as follows. 
The event $\cE_{s,v,h}$ occurs when an $h$-hop random walk from $s$
has length $2^{i+1}\ell$ and ends at $v$. 
Observe that $\prw{s}{v}{2^{i+1}\ell} = \sum_{h \leq 2^{i+1}\ell} \Pr[\cE_{s,v,h}]$ (because the number of hops is always
at most the length).  Since $\projw{s}{h}$ is a random walk vector in a symmetric Markov Chain,
the infinity norm is non-increasing in $h$. Thus, it suffices to find
a subset $S'' \subseteq S'$, $|S''| \geq |S'|/2$ such that
$\forall s \in S''$, $\exists v \in S, h \geq n^{\delta}$, $\Pr[\cE_{s,v,h}] \geq 1/n^{\delta(i+6)}$.

We define $U_s$ as given in \Clm{set}.
For all $v \in U_s$, by \Clm{set}, $\prw{s}{v}{2^{i+1}\ell} \geq 1/2n^{\delta i}$. 
Therefore, for all $v \in U_s$, 
\begin{equation} \label{eq:hops}
\sum_{h \geq 2^{i+1}\ell} \Pr[\cE_{s,v,h}] \geq 1/2n^{\delta i}
\end{equation}
We will construct $S''$ by finding $s$ where for some $v \in U_s$, $\sum_{h \leq n^{\delta}} \Pr[\cE_{s,v,h}]$ is sufficiently small.

\medskip

For any $h$,
$$ \frac{1}{|S|} \sum_{s \in S'} \sum_{v \in U_s} \Pr[\cE_{s,v,h}] (2^{i+1}\ell) \leq \EX_{W\sim \walkdist{h}}[\textrm{length of $W$}] = hn/|S|$$
Suppose $h \leq 2^{i+1}\ell/n^{4\delta}$. (This is true for all $h \leq n^{\delta}$). 
Then $\sum_{s \in S'} \sum_{v \in U_s} \Pr[\cE_{s,v,h}] \leq n^{1-4\delta}$,
and $\sum_{h \leq n^\delta} \sum_{s \in S'} \sum_{v \in U_s} \Pr[\cE_{s,v,h}] \leq n^{1-3\delta}$.

We rearrange to get
$$\sum_{s \in S'} \sum_{v \in U_s} \sum_{h \leq n^\delta} \Pr[\cE_{s,v,h}] \leq n^{1-3\delta}$$
% 
% 
% 		\noindent Let us define
% 		$$\trap_{h,\ell}(s) = \Pr[\text{an at most } h \text{-hop walk from $s$ ends in } U_s \text{ after exactly } 2^{i+1}\ell \text{ steps }].$$ 

By the Markov bound, there is a set $S'' \subseteq S'$, $|S''| \geq |S'|/2$
such that for all $s \in S''$,
$\sum_{v \in U_s} \sum_{h \leq n^\delta} \Pr[\cE_{s,v,h}]\leq 2n^{1-3\delta}/|S'|$.
By averaging, $\forall s \in S''$, $\exists v \in U_s$, such that 
$\sum_{h \leq n^{\delta}} \Pr[\cE_{s,v,h}] \leq 2n^{1-3\delta}/(|S'|\cdot|U_s|)$.
By the assumptions of \Lem{partition},  $|S'| \geq \alpha n \geq n^{1-\delta/2}$.
\Clm{set} bounds $|U_s| \geq n^{\delta(i-2)}$. Plugging these in,
% By the Markov bound, $|S''| \geq |S'|/2$.
% By averaging over $U_s$, we can assert the following. For every $s \in S''$, there exists some $v \in U_s$
% such that $\sum_{h \leq n^{\delta}} \Pr[\cE_{s,v,h}] \leq 2/(\alpha n^{3\delta} |U_s|)$.
\begin{equation}
 \sum_{h \leq n^{\delta}} \Pr[\cE_{s,v,h}] \leq \frac{2n^{1-3\delta}}{n^{1-\delta/2} n^{\delta(i-2)}} 
 \leq \frac{2}{n^{\delta(i+1/2)}}
%  2/(\alpha n^{\delta(i+1)}) \textrm{.} \label{eq:hops}
\end{equation}

Subtracting this bound from \Eqn{hops},
%  and using the fact that $\eps > \eps_{\sf CUTOFF}$ and $\alpha = \Omega(\eps / \log n)$, 
 $\sum_{h \in [n^\delta, 2^{i+1}\ell}] \Pr[\cE_{s,v,h}] \geq 1/4n^{\delta i}$.
By averaging, for some $h \in [n^\delta, 2^{i+1}\ell]$,
$\Pr[\cE_{s,v,h}] \geq 1/(2^{i+3}n^{\delta i}\ell) \geq 1/n^{\delta(i+6)}$. 
This completes the proof.
\end{proof}

We perform local partitioning on $M_S$, starting with arbitrary $s \in S''$.
We apply the Lov\'{a}sz-Simonovits curve technique. (The definitions are originally from~\cite{LS:90}.
Refer to Lecture 7 of Spielman's notes~\cite{Sp-notes} as well as Section 2 in Spielman-Teng~\cite{ST12}.) 
This requires a series of definitions.

\begin{itemize}
    \item Ordering of states at time $t$: At time $t$,
let us order the vertices in $M_S$ as $v^{(t)}_1, v^{(t)}_2, \ldots$ such that
$\projwp{s}{v^{(t)}_1}{t} \geq \projwp{s}{v^{(t)}_2}{t} \ldots$, breaking ties by vertex id.
    \item The LS curve $h_t$: We define a function $h_t:[0,|S|] \to [0,1]$ as follows. 
    For every $k \in [|S|]$, set $h_t(k) = \sum_{j \leq k} [\projwp{s}{v^{(t)}_j}{t}-1/|S|]$. (Set $h_t(0) = 0$.)
For every $x \in (k,k+1)$, we linearly interpolate to construct $h(x)$.
Alternately, $h_t(x) = \max_{\vec{w} \in [0,1]^{|S|}, \|\vec{w}\|_1 = x} \sum_{v \in S} [\projwp{s}{v}{t} - 1/n]w_i$.
    \item Level sets: For $k \in [0,|S|]$, we define the $(k,t)$-level set, $L_{k,t}$ to be $\{v^{(t)}_1, v^{(t)}_2, \ldots, v^{(t)}_k\}$.
    The \emph{minimum probability} of $L_{k,t}$ denotes $\projwp{s}{v^{(t)}_k}{t}$.
    \item Conductance: for some $T \subseteq S$ we define the conductance of $T$ in $M_S$ to be 
$$ \Phi(T) = \frac{\sum_{\substack{u \in T \\ v \in S \setminus T}} \projwp{u}{v}{1}}{\min(|T|, |S \setminus T|)}$$ 
\end{itemize}

The main lemma of Lov\'{a}sz-Simonovits is the following (Lemma 1.4 of \cite{LS:90}).

\begin{lemma} \label{lem:ls} For all $k$ and all $t$,
$$ h_t(k) \leq \frac{1}{2}[h_{t-1}(k - 2\min(k,n-k)\Phi(L_{k,t})) + h_{t-1}(k + 2\min(k,n-k)\Phi(L_{k,t}))]$$
\end{lemma}

\iffalse

We summarize how this lemma is usually put to work (also refer to Theorem 7.3.3 of 
Lecture 7) in~\cite{Sp-notes}). This is setup so that the function $h_t$ tracks how close
at step $t$ the distribution is from the stationary distribution. At stationarity,
the $h_t$ curve is zero everywhere. The key is that this curve keeps ``dropping'' 
as the walk progresses. In almost all applications of this machinery, one iterates
the above lemma to get a statement of the form 

$$h_t(x) \leq \sqrt{x}(1 - O(\phi(S)^2))^t + O(x).$$

\noindent for some set $S$ that befits the application. The lemma below uses the LS 
machinery in the same way by putting the following intuition to work. Suppose we are
promised that any massive level set (i.e., a level set with large minimum probability), 
within first $t-1$ steps, has large conductance. Then at the $t$-th step the LS curve 
$h_t$ drops steeply everywhere. To formalize this, it is helpful to define

$$\cL_t = \{L_{k,t} : k \in [n] \text{ is such that minimum probability of } L_{k,t} \geq 1/10n^{\delta(i+6)}\}$$

\noindent as the set of all \emph{massive level sets} at the $t$-th step. 
\fi

The typical use of the Lov\'{a}sz-Simonovitz technique is to argue about rapid mixing when
all conductances (or conductances of sufficiently large sets) are lower bounded.
We consider a scenario in which only sets with minimum probability at least (say) $p$
have high conductance. In this case, we can guarantee that the largest probability
will converge to $p$.

\begin{lemma} \label{lem:level} Suppose the following holds.
For all $t' \leq t$, if the minimum probability of $L_{k, t'}$ is at least $1/10n^{\delta(i+6)}$, then
$\Phi(L_{k,t'}) \geq n^{-\delta/4}$,  
Then, $\forall x \in [0,n]$, $h_t(x) \leq \sqrt{x}(1-n^{-\delta/2}/4)^t + x/10n^{\delta(i+6)}$.
\end{lemma}

%%%%%%%%%%%%%%
%%%%%TODO%%%%%%%
%\begin{lemma} \label{lem:level} For all $t' \leq t$, if:
%		All level sets $L_{k,t'}$ with large minimum probability at (least $1/10n^{\delta(i+6)}$)
%		have large conductance -- that is, $\Phi(L_{k,t'}) \geq n^{-\delta/4}$.
%Then for all $x \in [0,n]$, $h_t(x) \leq \sqrt{x}(1-n^{-\delta/2}/4)^t + x/10n^{\delta(i+6)}$.
%\end{lemma}
%%%%%%%%%%%%%%%%%%%%%%%%%%

	\begin{proof}

			Notice that it suffices to show this claim for integral values of $x$ since $h_t$ is concave. To begin with, note that if $x = k \geq n^{\delta(i+6)}$, then 
			the RHS is at least 1. Thus the bound is trivially true. Let us assume that 
			$k < n^{\delta(i+6)} < n/2$. We proceed by induction over $t$ and split into 
			two cases based on the conductance of level sets.\\ 

Suppose $k$ is such that $\Phi(L_{k,t}) \geq n^{-\delta/4}$. By \Lem{ls} and concavity of $h$, we have the following at $x = k$
\begin{align}
    h_t(k) & \leq  \frac{1}{2}\left( h_{t-1}(k(1-2n^{-\delta/4})) + h_{t-1}(k(1+2n^{-\delta/4})) \right) \\
    & \leq  \frac{1}{2}\left( \sqrt{k(1-2n^{-\delta/4})}(1-n^{-\delta/2}/4)^{t-1} + \sqrt{k(1+2n^{-\delta/4})}(1-n^{-\delta/2}/4)^{t-1} + \frac{2k}{10n^{\delta(i+6)}}\right)  \\
    & \leq  \frac{1}{2}\left(\sqrt{k}(1-2n^{-\delta/4})^{t-1}(\sqrt{1-2n^{-\delta/4}} + \sqrt{1+2n^{-\delta/4}}) + \frac{2k}{10n^{\delta(i+6)}}\right)\\
    & \leq  \sqrt{k}(1-n^{\delta/2}/2)^t + k/n^{\delta(i+6)}
\end{align}
		\noindent For the last inequality we use the bound $\left( \sqrt{1 + x} + \sqrt{1 - z}\right) / 2 \leq 1 - z^2/8$.

Now, consider the case where $k$ is such that  $\Phi(L_{k,t}) \leq n^{-\delta/4}$. By assumption, it must be that $L_{k, t'}$ must have minimum probability less than $1/10n^{\delta(i + 6)}$. Let $k'$ be the largest integer less than $k$ such that $\Phi(L_{k', t}) \geq n^{-\delta/4}$. By the previous case, $h_t(k') \leq \sqrt{k'}(1-n^{\delta/2}/2)^t + k/n^{\delta(i+6)}$.
Using this and the concavity of $h_t$, we get
\begin{align}
	h_t(k) &\leq h_t(k') + (k-k')/10n^{\delta(i+6)} \\
	&\leq \sqrt{k'}(1-n^{-\delta/2}/2)^t + k'/10n^{\delta(i+6)} + (k-k')/10n^{\delta(i+6)} \\
	&\leq \sqrt{k}(1-n^{-\delta/2}/2)^t + k/10n^{\delta(i+6)}
\end{align}
\end{proof}

\subsection{Proof of \Lem{partition}} \label{sec:proof:of:pl}

\begin{proof} Define $S''$ as given in \Lem{good}.
For any $s \in S''$, $\|\projw{s}{n^\delta}\|_\infty \geq 1/n^{\delta(i+6)}$.
By the definition of the LS curve, $h_{n^\delta}(1) \geq 1/n^{\delta(i+6)}$.
Suppose (for contradiction's sake) all level sets for $t \leq n^\delta$ with minimum probability at least $1/10n^{\delta(i+6)}$
have conductance at least $n^{-\delta/4}$. By \Lem{level}, $h_{n^\delta}(1) \leq (1-n^{-\delta/2}/4)^{n^\delta} + 1/10n^{\delta(i+6)}$
$< 1/n^{\delta(i+6)}$. This contradicts the bound obtained by \Lem{good}.

Thus, for every $s \in S''$, there exists some level set for $t_s \leq n^\delta$ with
minimum probability at least $1/10n^{\delta(i+6)}$ and conductance $<n^{-\delta/4}$.
Let us call this level set $P_s$. We also use the fact that $|P_s| < |S|/2$.
By the construction of $M_S$, we have, 
$$\Phi(P_s) \geq \frac{\sum_{\substack{x \in P_s \\ y \in S \setminus P_s}}\projwp{x}{y}{1}}{\min(|P_s|,|S\setminus P_s|)}
 = \frac{E(P_s, S\setminus P_s)}{2d|P_s|} $$
The first inequality follows because we restrict the numerator to length one transitions in the Markov Chain $M_S$ (which correspond 
to edges in $G$). Rearranging, we get $E(P_s, S \setminus P_s) \leq n^{-\delta/4} (2d|P_s|)$.

For all $s \in S''$ and $v \in P_s$, $\projwp{s}{v}{n^\delta} \geq 1/10n^{\delta(i+6)}$. 
Set $L = 160n^{\delta(i+7)}/\alpha$. Let $\widetilde{S}$ be the subset of $S''$
such that $\forall s \in \widetilde{S}$, $P_s$ is such that $\forall v \in P_s$, $\sum_{l \leq L} \prw{s}{v}{l} \geq 1/20n^{\delta(i+6)}$. 
By averaging, $\exists l \leq L$ such that $\prw{s}{v}{l} \geq \alpha/n^{\delta(2i+14)}$.

We have seen that $\widetilde{S}$ satisfies the two desired properties: for all $s \in \widetilde{S}$ $E(P_s, S \setminus P_s)\leq 2 n^{-\delta/4}d|P_s|/\alpha$ and for all $v \in P_s$, $\exists t \leq 160n^{\delta(i+7)}$ such that $p_{s, t}(v) \geq \alpha/n^{\delta(2i + 14)}$. 
It only remains to prove a lower bound on size, or alternately, an upper bound on $|S''\setminus \widetilde{S}|$.

Consider any $s \in S'' \setminus \widetilde{S}$. There exists some $v_s \in P_s$
such that $\projwp{s}{v_s}{n^\delta} \geq 1/10n^{\delta(i+6)}$ but $\sum_{l \leq L} \prw{s}{v_s}{l} < 1/20n^{\delta(i+6)}$.
Let us use $\prwhop{s}{v_s}{l}$ to denote the probability of reaching $v_s$ from $s$ in an $l$-length walk
that makes $n^\delta$ hops. Observe that 
\begin{align}
\projwp{s}{v_s}{n^\delta} = \sum_{l \geq n^\delta} \prwhop{s}{v_s}{l}
= \sum_{l = n^\delta}^L \prwhop{s}{v_s}{l} + \sum_{l > L} \prwhop{s}{v_s}{l} 
& \leq \sum_{l = n^\delta}^L \prw{s}{v_s}{l} + \sum_{l > L} \prwhop{s}{v_s}{l} \\
& < 1/20n^{\delta(i+6)} + \sum_{l > L} \prwhop{s}{v_s}{l}
\end{align}
The last inequality follows from the fact that $s \in S'' \setminus \widetilde S$, and hence $\sum_{l = n^\delta}^L \prw{s}{v_s}{l} <  1/20n^{\delta(i+6)} $. Since $\projwp{s}{v_s}{n^\delta} \geq 1/10n^{\delta(i+6)}$, the above calculation shows that $\sum_{l > L} \prwhop{s}{v_s}{l} > 1/20n^{\delta(i+6)}$. 
Thus,
\begin{equation}
\frac{1}{|S|} \sum_{s \in S''\setminus \widetilde{S}} \sum_{l > L} \prwhop{s}{v_s}{l} L > \frac{|S'' \setminus \widetilde{S}|\cdot L}{|S| 20 n^{\delta(i+6)}}
= \frac{160 \alpha^{-1} n^{\delta(i+7)}\cdot|S'' \setminus \widetilde{S}|}{20 |S| n^{\delta(i+6)}} = \frac{8 n^\delta |S'' \setminus \widetilde{S}|}{\alpha |S|}
\end{equation}
By \Lem{kac}, 
\begin{equation}
\frac{1}{|S|} \sum_{s \in S''\setminus \widetilde{S}} \sum_{l > L} \prwhop{s}{v_s}{l} L \leq \EX_{W \sim \walkdist{n^\delta}} [\textrm{length of $W$}] = \frac{n^{1+\delta}}{|S|}
\end{equation}
Combining the above, $|S'' \setminus \widetilde{S}| \leq \alpha n/8$. By \Lem{good}, $|S''| \geq |S'|/2 \geq \alpha n/2$,
yielding the bound $|\widetilde{S}| \geq \alpha n/4$.

\end{proof}

\section{Wrapping it all up: the proof of \Thm{main-result}} \label{sec:final}

We have all the tools required to complete the proof of \Thm{main-result}. 
Our aim is to show that if \isminorfree$(G,\eps,H)$ outputs an $H$-minor with probability $<2/3$,
then $G$ is $\eps$-close to being $H$-minor free. 
Henceforth in this section, we will simply assume the ``if" condition.

The following decomposition procedure is used by the proof.
We set parameter $\alpha = \eps/(50r^4\log n)$.

\medskip

\noindent
\fbox{
	\begin{minipage}{0.9\textwidth}
		{\tt \decompose $(G)$}
		
		\smallskip
		\begin{compactenum}
			\item Initialize $S = V$ and $\cP = \emptyset$.
			\item For $i = 1, \ldots, 5r^4$:
			\begin{compactenum}
				\item Assign $S' := \left\{s \in S : ||\trunvec{s}{S}{i+1}||_2^2 \geq 1/n^{\delta i}\right\}$
				\item While $|S'| \geq \alpha n$:
				\begin{compactenum}
					\item Choose arbitrary $s \in S''$, and let $P_s$ be as in \Lem{partition}.
                    \item Add $P_s$ to $\cP$ and assign $S := S \setminus P_s$
					\item Assign $S' := \left\{s \in S : ||\trunvec{s}{S}{i+1}||_2^2 \geq 1/n^{\delta i}\right\}$
				\end{compactenum}
				\item Assign $S := S \setminus S'$
				\item Assign $X_i := S'$
			\end{compactenum}
            \item Let $X = \bigcup_i X_i$.
            \item Output the partition $\cP, X, S$
		\end{compactenum}
\end{minipage}}\\

The procedure \decompose{} repeatedly employs \Lem{partition} for values of $i \leq 5r^4$. 
In the $i$th iteration, eventually $|S'|$ becomes too small for \Lem{partition}.
Then, $S'$ is moved (from $S$) to an ``excess" set $X_i$, and the next iteration begins.
\decompose{} ends with a partition $\cP, X, S$ where each set in $\cP$ is a low conductance cut,
$X$ is fairly small, and \findclique{} succeeds with high probability on every vertex in $S$.

This is formalized in the next lemma. 

\begin{lemma} \label{lem:decompose} Assume $\eps > \cutoff$. Suppose \isminorfree$(G,\eps,H)$ outputs an $H$-minor
with probability $< 2/3$. Then, the output of \decompose{} satisfies the following conditions.
\begin{asparaitem}
    \item $|X| \leq \eps n/10$.
    \item $|S| \leq \eps n/10$.
    \item $\forall P_s \in \cP, v \in P_s$, $\exists t \leq 160 n^{6 \delta r^4}/\alpha$ such that $\prw{s}{v}{t} \geq \ballprob$.
    \item There are at most $\eps n/10$ edges that go between different $P_s$ sets.
\end{asparaitem}
\end{lemma}

\begin{proof} Consider the $X_i$'s formed by \decompose. Each of these has size at most $\alpha n = \eps n/50r^4\log n$, and there are at most $5r^4$ of these. 
Clearly, their union has size at most $\eps n/10$.

The third condition holds directly from \Lem{partition}.
Consider the number of edges that go between $P_s$ and the rest of $S$, when $P_s$ was
constructed (in \decompose). By \Lem{partition} again, the number of these edges is at most $2n^{-\delta/4} d|P_s|/\alpha
= 40r^4 (\log n) \eps^{-1} n^{-\delta/4} d|P_s|$. Note that $\eps > \cutoff$. For sufficiently small constant $\delta$,
the number of edges between $P_s$ and $S\setminus P_s$ (at the time of removal) is at most $\eps |P_s|/10$. 
The total number of such edges is at most $\eps n/10$ (since $P_s$ are all disjoint).
% Observe that any edge that goes between different $P_s$ sets must be such a cut edge, proving the fourth condition.

% The second condition on $|S|$ is where we require the behavior of \isminorfree.
Suppose, for contradiction's sake, that $|S| > \eps n/10$. 
Consider the stratification process with $R_0 = S$. By construction of $S$, $\forall s \in S$,
$||\trunvec{s}{S}{5r^4+1}|| \leq 1/n^{5\delta r^4}$. Thus, all of these
vertices will lie in strata numbered $5r^4$ or above. Since $\eps > \cutoff$, by \Lem{strata},
at most $\eps n/\log n$ vertices are in strata numbered more than $1/\delta + 3$.
By \Thm{findclique_succeeds_probability}, for at least $\eps n/10 - \eps n/\log n \geq \eps n/20$
vertices, the probability that the paths discovered by \findclique$(s)$ contain
a $K_{r^2,r^2}$-minor is at least $n^{-4\delta r^4}$. Since a $K_{r^2,r^2}$ minor contains
an $H$-minor, the algorithm (in this situation) will succeed in finding an $H$-minor.

All in all, this implies that the probability that a single call to \findclique{} finds an $H$ minor
is at least $n^{-5\delta r^4}$. Since \isminorfree{} makes $n^{20\delta r^4}$ calls to \findclique{},
an $H$-minor is found with probability at least $5/6$.
This is a contradiction, and we conclude that $|S| \leq \eps n/10$.

% 
% We can see that all members of $S$ are in $S_i$ for $i \geq T$ since for all $i$, $R_i \subseteq S$ and at the end of \decompose $||\trunvec{s}{S}{i+1}|| \leq 1/n^{\delta i}$ for all $s \in S$ and $ i \leq T$ and hence . By \Lem{strata}, at most $\frac{\eps n}{\log n}$ vertices of $S$ are not in some $S_i$ for $T \leq i \leq 1/\delta + 3$. Therefore, the probability that we choose a vertex $s \in V$ u.a.r. which is in some $S_i$ for $T \leq i \leq 1/\delta + 3$ is at least
% \begin{equation}\label{eq:prob_vtx_in_high_stratum}
% 	\left(\frac{\eps n}{10} - \frac{\eps n}{\log n}\right) / n \geq \frac{\eps}{20}
% \end{equation}
% \noindent for sufficiently large $n$. Multiplying \Eqn{prob_vtx_in_high_stratum} by the probability bound in \Thm{findclique_succeeds_probability}, we get that the probability (over the randomness in the choice of $s$ and random walks in \findclique) that a single call to \findclique$(s)$ rejects is at least
% $ \frac{\eps}{20} n^{-4\delta r^4}\textrm{,} $
% and the the probability that none of the $\rho_2$ calls to \findclique reject is at most
% \begin{equation}
% 	\left( 1 - \frac{\eps}{20} n^{-4\delta r^4}\right) ^{\rho_2} \leq \exp \left( - \rho_2 \frac{\eps}{10} n^{-4\delta r^4}\right) < 1/3 \textrm{.}
% \end{equation}
% 
\end{proof}

% \begin{claim} \label{clm:ball} Assume $\eps > \cutoff$. Suppose \isminorfree$(G,\eps,H)$ outputs an $H$-minor
% with probability $< 2/3$. Then, for at most $n^{1-30\delta r^4}$ vertices $s$, $B_s$ induces an $H$-minor.
% \end{claim}
% 
% \begin{proof} Suppose not. Observe that \localsearch{} is called on at least $n^{35\delta r^4}$ vertices,
% so with probability at least $1 - (1-n^{-30 \delta r^4})^{n^{35\delta r^4}}$ $> 5/6$, the algorithm samples vertex $s$ where $B_s$ induces an $H$-minor.
% There are at most $160\alpha^{-1} n^{6\delta r^4} \cdot \alpha^{-1} n^{11\delta r^4} \leq n^{18\delta r^4}$ vertices
% in $B_s$. The probability of ending at such a vertex $v$ (from $s$) is at least $\alpha n^{-11\delta r^4} \geq n^{-12\delta r^4}$.
% Since \localsearch$(s)$ performs $n^{30\delta r^4}$ random walks from $s$, the probability of not adding
% $v$ to $B$ is at most $\exp(n^{15\delta r^4})$. Taking a union bound over all of $B_s$,
% with probability at least $5/6$, the set $B$ discovered by \localsearch$(s)$ is a superset of $B_s$.
% In this situation, \localsearch$(s)$ finds an $H$-minor. Thus, the overall probability of outputting an $H$-minor
% is at least $(5/6)^2 > 2/3$, a contradiction.
% \end{proof}
% 
And now, we can prove the correctness guarantee of \isminorfree.

\begin{claim} \label{clm:correct} Suppose \isminorfree$(G,\eps,H)$ outputs an $H$-minor with probability $< 2/3$.
Then $G$ is $\eps$-close to being $H$-minor free.
\end{claim}

\begin{proof} If $\eps \leq \cutoff$, then \isminorfree{} runs an exact procedure. So the claim is clearly true.
Henceforth, assume $\eps > \cutoff$. Apply \Lem{decompose} to partition $V$ into $\cP, X, S$.

Call $s \in V$ bad, if there is a corresponding $P_s \in \cP$ and $P_s$ induces an $H$-minor. 
By \Lem{decompose}, for all $v \in P_s$, $\exists t \leq 160 n^{6 \delta r^4}/\alpha$ such that
$\prw{s}{v}{t} \geq \alpha/n^{11\delta r^4}$. 
Note that $160 n^{6\delta r^4}/\alpha \leq n^{7\delta r^4}$
and $\alpha/n^{11\delta r^4} \geq n^{-12\delta r^4}$. Also, $|P_s| \leq 160 (n^{6\delta r^4}/\alpha) \times (n^{11\delta r^4}/\alpha) 
\leq n^{18\delta r^4}$.
Note that \localsearch$(s)$ performs walks of all lengths up to $n^{7\delta r^4}$,
and performs $n^{30\delta r^4}$ walks of each length. For any $v \in P_s$, the probability that \localsearch$(s)$
does not add $v$ to $B$ (the set of ``discovered" vertices in \localsearch$(s)$) is at most $(1-n^{-12\delta r^4})^{n^{30\delta r^4}}
\leq 1/n^2$. Taking a union bound over $P_s$, the probability that $P_s$ is not contained in $B$
is at most $1/n$.  Consequently, for bad $s$, \localsearch$(s)$ outputs an $H$-minor with probability $>1-1/n$.

Suppose there are more than $n^{1-30\delta r^4}$ bad vertices. The probability that a uar
$s \in V$ is bad is at least $n^{-30\delta r^4}$. Since \isminorfree$(G,\eps,H)$ invokes
\localsearch{} $n^{35\delta r^4}$ times, the probability that \localsearch$(s)$
is invoked for a bad vertex is at least $1-1/n$. Thus, \isminorfree$(G,\eps,H)$ 
outputs an $H$-minor with probability $>1-2/n$, contradicting the claim assumption.

We conclude that there are at most $n^{1-30\delta r^4}$ bad vertices. Each $P_s$
has at most $n^{18\delta r^4}$ vertices, and $|\bigcup_{s \ \textrm{bad}} P_s| \leq n^{1-12\delta r^4} \leq \eps n/10$.

% Similarly, call $B_s$ bad if it induces an $H$-minor.
% Observe that for all $P_s \in \cP$, $P_s \subseteq B_s$.
% Thus, the union of bad $P_s$ sets is contained in the union on bad $B_s$ sets.
% The size of any $B_s$ is at most $160\alpha^{-1} n^{6\delta r^4} \times \alpha^{-1} n^{11\delta r^4} \leq n^{18\delta r^4}$.
% There are at most $n^{1-30\delta r^4}$ bad $B_s$ sets, by \Clm{ball}.
% So, the total union is at most $n^{1-12\delta r^4} \leq \eps n/10$.
% 
We can make $G$ $H$-minor free by deleting all edges incident to $X$, all edges
incident to $S$, all edges incident to vertices in any bad $P_s$ sets, and all edges
between $P_s$ sets. By \Lem{decompose} and the bound given above,
the total number of edges deleted is at most $4\eps dn/10 < \eps dn$.
\end{proof}

Finally, we bound the running time.

\begin{claim} \label{clm:runtime} The running time of \isminorfree$(G,\eps,H)$ is $\runtime$.
\end{claim}

\begin{proof} If $\eps < \cutoff$, then the running time is simply $O(n^2)$. 
Since $\eps < n^{-\delta/\exp(2/\delta)}$, this can be expressed as $\eps^{-2\exp(2/\delta)/\delta}$.

Assume $\eps \geq \cutoff$. The total number of vertices encountered by all the \localsearch{} calls is $n^{O(\delta r^4)}$.
There is an extra $d$ factor to determine all incident edges, through vertex queries.
Thus, the total running time is $d n^{O(\delta r^4)}$, because of the quadratic overhead of \RS.
Consider a single iteration for the main loop of \findclique. First, \findclique{} performs $2r^2$ random walks of length $2^{i + 1} n^{5\delta}$, and then for each of these, \findpath{} performs $n^{\delta i/2 + 9\delta}$ walks of length $2^i n^{5\delta}$. Hence, the total steps (and thus, queries) in all walks performed by a single call to \findclique{} is 
\begin{align}
\sum_{i = \threshold }^{1/\delta + 3} \left(2r^2 2^{i + 1} n^{5\delta}  + 2r^2 n^{\delta i/2 + 9\delta}2^{i} n^{5\delta}\right) =r^2 n^{1/2 + O(\delta)}\textrm{.}
\end{align}
While this is the total number of vertices encountered, we note that the calls made to \RS$(F,H)$ are for 
much smaller graphs. The output of find path has size $O(2^{1/\delta} n^{5\delta})$, and the subgraph $F$
constructed has at most $O(2^{1/\delta} n^{5\delta})$ vertices. We incur an extra $d$ factor to determine
the induced subgraph, through vertex queries. Thus, the time for each call to \RS$(F,H)$
is $n^{O(\delta)}$. There are $n^{O(\delta r^4)}$ calls to \findclique, and we can bound
the total running time by $d n^{1/2 + O(\delta r^4)}$.
\end{proof}

\section*{Acknowledgements} We would like to acknowledge Madhur Tulsiani for his improvements to \Lem{trun}.

\bibliographystyle{alpha}
\bibliography{minor-freeness}

\appendix

\end{document}

%% file: preamble.tex
	\usepackage{a4,geometry}
\usepackage{graphicx}
\usepackage{amsmath,amssymb,amsthm,mathtools}
\usepackage{paralist}
\usepackage{bm}
\usepackage{xspace}
\usepackage{url}
\usepackage{fullpage, prettyref}
\usepackage{boxedminipage}
\usepackage{wrapfig}
\usepackage{ifthen}
\usepackage{color}
\usepackage{xcolor}
\usepackage{framed}
\usepackage{algorithmic,algorithm}
\usepackage[pagebackref,letterpaper=true,colorlinks=true,pdfpagemode=none,urlcolor=blue,linkcolor=blue,citecolor=violet,pdfstartview=FitH]{hyperref}
\usepackage{fullpage}

\usepackage{thmtools}
\usepackage{thm-restate}
\newtheorem{theorem}{Theorem}[section]

\newtheorem{lemma}[theorem]{Lemma}
\newtheorem{claim}[theorem]{Claim}
\newtheorem{corollary}[theorem]{Corollary}
\newtheorem{definition}[theorem]{Definition}
\newtheorem{proposition}[theorem]{Proposition}

\newcommand{\ignore}[1]{}

\newcommand{\cC}{{\cal C}}
\newcommand{\cD}{\mathcal{D}}
\newcommand{\cE}{{\cal E}}
\newcommand{\cF}{\mathcal{F}}

\newcommand{\cL}{{\cal L}}

\newcommand{\cP}{\mathcal{P}}
\newcommand{\cQ}{\mathcal{Q}}

\newcommand{\cS}{\mathcal{S}}

\newcommand{\cW}{{\cal W}}

\newcommand{\R}{\mathbb R}

\newcommand{\eps}{\varepsilon}

\newcommand{\poly}{\mathrm{poly}}

\newcommand{\bone}{{\bf 1}}

\newcommand{\bx}{\boldsymbol{x}}

\newcommand{\bH}{\boldsymbol{H}}

\newcommand{\bW}{\boldsymbol{W}}

\newcommand{\NN}{\mathbb{N}}

\newcommand{\EX}{\hbox{\bf E}}
\newcommand{\var}{\hbox{\bf var}}

\newcommand{\otilde}{\widetilde{O}}

\newcommand{\Sec}[1]{\hyperref[sec:#1]{\S\ref*{sec:#1}}} %section
\newcommand{\Eqn}[1]{\hyperref[eq:#1]{(\ref*{eq:#1})}} %equation
\newcommand{\Fig}[1]{\hyperref[fig:#1]{Fig.\,\ref*{fig:#1}}} %figure
\newcommand{\Tab}[1]{\hyperref[tab:#1]{Tab.\,\ref*{tab:#1}}} %table
\newcommand{\Thm}[1]{\hyperref[thm:#1]{Theorem\,\ref*{thm:#1}}} %theorem
\newcommand{\Fact}[1]{\hyperref[fact:#1]{Fact\,\ref*{fact:#1}}} %fact
\newcommand{\Lem}[1]{\hyperref[lem:#1]{Lemma\,\ref*{lem:#1}}} %lemma
\newcommand{\Prop}[1]{\hyperref[prop:#1]{Prop.~\ref*{prop:#1}}} %property
\newcommand{\Cor}[1]{\hyperref[cor:#1]{Corollary~\ref*{cor:#1}}} %corollary
\newcommand{\Conj}[1]{\hyperref[conj:#1]{Conjecture~\ref*{conj:#1}}} %conjecture
\newcommand{\Def}[1]{\hyperref[def:#1]{Definition~\ref*{def:#1}}} %definition
\newcommand{\Alg}[1]{\hyperref[alg:#1]{Alg.~\ref*{alg:#1}}} %algorithm
\newcommand{\Ex}[1]{\hyperref[ex:#1]{Ex.~\ref*{ex:#1}}} %example
\newcommand{\Clm}[1]{\hyperref[clm:#1]{Claim~\ref*{clm:#1}}} %example
\newcommand{\Step}[1]{\hyperref[step:#1]{Step~\ref*{step:#1}}} %example